\documentclass{article}
\usepackage[utf8]{inputenc}\usepackage{newunicodechar}\include{chars}

\usepackage{pdftexcmds} %
\usepackage{fnpct}
\setfnpct{add-punct-marks=!?:;} %

\usepackage{ellipsis}

\usepackage{fancyvrb}

\usepackage[normalem]{ulem} %
\usepackage{pifont} %
\usepackage{xcolor}
\usepackage{fancyhdr}
\usepackage{xspace}

\usepackage{cite}

\usepackage{graphicx} %

\usepackage{standalone}

\usepackage{tikz}\usetikzlibrary{shapes,calc,arrows,fit,positioning,decorations.pathmorphing,snakes,intersections,shapes.geometric,trees,cd}

\usepackage[hyphens]{url} %

\usepackage{afterpage}
\usepackage{wrapfig}
\usepackage{enumitem}
\usepackage{alphalph} %

\usepackage{booktabs} %

\newenvironment{tab}[1]
{%
\renewcommand{\arraystretch}{1.2} %
\begin{tabular}{@{}#1@{}}
\toprule
}
{\bottomrule
\end{tabular}}
\usepackage{array}
\newcolumntype{L}[1]{>{\raggedright\let\newline\\\arraybackslash\hspace{0pt}}m{#1}}
\newcolumntype{C}[1]{>{\centering\let\newline\\\arraybackslash\hspace{0pt}}m{#1}}
\newcolumntype{R}[1]{>{\raggedleft\let\newline\\\arraybackslash\hspace{0pt}}m{#1}}

\usepackage{tabularx} %

\usepackage{environ}

\NewEnviron{tabx}[1]{\renewcommand{\arraystretch}{1.2} %
\tabularx{\textwidth}{@{}#1@{}}
\toprule
\BODY
\bottomrule
}[\endtabularx]

\usepackage{multirow}

\makeatletter
\newcommand{\ensurepunct}[1]{\ifnum\spacefactor>\@m\relax\else#1\fi}
\let\Paragraph=\paragraph
\renewcommand{\paragraph}[1]{\Paragraph{#1\ensurepunct{.}}}
\makeatother

 \def\final{}

\ifdefined\final
  \newcommand{\authornote}[2]{}
\else
  \newcommand*{\authornote}[2]{\textcolor{#1}{[#2]}}
\fi

\newcommand*{\mct}   [1]{\authornote{magenta}{MCT: {#1}}}

\usepackage{amsmath}
\usepackage{amsfonts}
\usepackage{amssymb}

\usepackage{amsthm}

\theoremstyle{definition}
\newtheorem{thm}     {Theorem}
\newtheorem{lem}[thm]{Lemma}
\newtheorem{prp}[thm]{Proposition}
\newtheorem{cor}[thm]{Corollary}

\newtheorem{dfn}[thm]{Definition}

\newtheorem{assn}[thm]{Assumption}

\newtheoremstyle{taskstyle}%
  {.32\baselineskip±.15\baselineskip}%
  {.32\baselineskip±.15\baselineskip}%
  {\it}%
  {}%
  {\bf}%
  {. }%
  { }%
  {}%
\theoremstyle{taskstyle}

\newcounter{subtasknumber}
\newcounter{subtasklabel}
\renewcommand{\thesubtasklabel}{\textbf{\thetasknumber}}

\newcommand{\mc}{\mathcal}
\newcommand{\msf}{\mathsf}

\newcommand*{\set}[2]{\{{#1}~|~{#2}\}}

\DeclareMathOperator{\st}{\,s.t.\,}

\newcommand{\given}{\mathrel{|}}
\DeclareMathOperator{\Cr}{Cr}
\DeclareMathOperator{\Fr}{Fr}
\newcommand{\inter}[1]{\mathtt{#1}}
\newcommand{\preisub}[2]{[{#1}{\rightarrow}#2]}
\newcommand{\isub}[2]{\preisub{\inter{#1}}{#2}}
\newcommand{\equals}{{=}}
\newcommand{\doo}{\mathrm{do}}

\usepackage[T1]{fontenc}

\usepackage{stmaryrd}

\newcommand{\Expect}{{\rm I\kern-.3em E}}

\newcommand{\prog}{s}

\usepackage{bbm} %
\newcommand{\cellcolor}[1]{\relax}

\newcommand{\changes}[1]{\underline{#1^\cdot}}
\newcommand{\both}[2]{\{{#1}, {#2}\}}

\begin{document}
\hyphenation{dis-crim-i-nat-ion non-dis-crim-i-nat-ion} %

\title{\Large Correspondences between Privacy and Nondiscrimination\\ \large Why They Should Be Studied Together}
\author{Anupam Datta\\ CMU \and Shayak Sen\\ CMU \and Michael Carl Tschantz\\ ICSI}

\maketitle

\begin{abstract}
Privacy and nondiscrimination are related but different.  We make this
observation precise in two ways.  First, we show that both privacy and
nondiscrimination have two versions, a causal version and a statical
associative version, with each version corresponding to a competing view of the
proper goal of privacy or nondiscrimination.
Second, for each version,  
we show that a difference between the privacy edition of the version and
the nondiscrimination edition of the version 
is related to the difference between Bayesian
probabilities and frequentist probabilities.
In particular, privacy admits both Bayesian and frequentist interpretations
whereas nondiscrimination is limited to the frequentist interpretation.
We show how the introduced correspondence allows results from one area
of research to be used for the other.
\end{abstract}

\section{Introduction}

Privacy and nondiscrimination appear both related and yet different.  The two
are both norms that help to ensure that powerful entities treat
people fairly by not violating social values.
However, privacy is typically seen as protecting information from
disclosure while nondiscrimination is seen as prohibiting certain
behaviors based upon known information about protected attributes,
such as gender or race.

In this paper, we explore the relationship between privacy and
nondiscrimination from a technical angle. We demonstrate that key aspects of
privacy and nondiscrimination mirror each other at a formal level, and
make a case for the exchange of techniques between the communities
independently studying the two.
Further, for each norm, there exist two stances within the research
community studying it: whether the norm, privacy or nondiscrimination,
should be measured in terms of association (probabilistic dependence)
or in terms of causation.
We express, in terms of probability theory,
the two stances for each of the two norms.
Doing so makes precise the exact difference between each pair of
stances,
revealing that each pair differs over the same issue.

In more detail, for privacy, the competing stances lead to
what we will call \emph{associative privacy} and \emph{causal privacy}.
Associative privacy properties typically show up in works
attempting to minimize
the knowledge gained by an adversary upon observing the outcomes of a
computation,
that is, works attempting to provide \emph{statistical nondisclosure}
(e.g.,~\cite{dalenius77statistik}).
Causal privacy instead focuses on whether some particular action
leads to a large change, with differential privacy being the prime
example
(e.g.,~\cite{warner65asa,dwork06crypto,dwork06icalp,bassily13focs,kasiviswanathan14jpc,mcsherry16blog1,mcsherry16blog2}).
Tschantz~et~al.\@ has already noted the causal nature
of differential privacy~\cite{tschantz17arxiv}.

As for nondiscrimination, two stances predominate in U.S. law,
\emph{disparate impact} and \emph{disparate treatment},
with similar counterparts in the rest of the world.
These are complex legal tests involving concepts that are difficult
to apply to algorithms, such as intent, and numerous caveats and exceptions.
However, at their cores are two standards, which we will call
\emph{associative nondiscrimination} and \emph{causal nondiscrimination}.
Associative nondiscrimination demands that
members of protected classes should not disproportionately suffer adverse actions.
Roughly speaking, to avoid a finding of associative discrimination, or disparate impact, a governed entity (e.g., a large employer) must ensure that the proportion of members of a protected class (e.g., women or a minority race) experiencing some adverse action (e.g., firing or not hiring) is roughly equal to the proportion of non-members experiencing it.
Courts have provided statistical characterizations of disparate impact in a number of settings, such as the 80\% standard in employment~\cite{burger71scotus}.
Causal nondiscrimination demands protected attributes do not cause people to experience some adverse action, which is the core of disparate treatment.
In short, to win a case under disparate treatment, the plaintiff must show, among other things, that the defendant subjected the plaintiff to some adverse action \emph{because} of the plaintiff's status as a member of a protected class.
For disparate treatment, the courts have used a common-sense approach that looks at motivations.
From it, we extract a mathematical characterization that focuses on just casual processes, ignoring motivations, along the lines of Pearl's treatment of the issue~\cite{pearl09book}.

Table~\ref{tab:properties-summary-words} provides a summary of the relationships between privacy and nondiscrimination.
It shows that properties for each norm vary along a binary axis falling into one of two stances: an associative or causal one.
We will call this axis the \emph{dependence axis} since association and causation are two forms of \emph{dependence} between random variables.
(``Statistical independence'' means a lack of association.  People often speak of ``causal dependence''.  Sometimes the word is used ambiguously as in ``the dependent variable''.)
It further shows that for each nondiscrimination property there is a corresponding privacy property.

\begin{table}
\centering
\renewcommand{\arraystretch}{1.4} %
\begin{tabular}{| l | l | l |}
  \hline
                      &  Associative &  Causal   \\
  \hline
Nondiscrimination
& Disparate impact
& Disparate treatment
\\
\hline
Privacy
&  Statistical nondisclosure
&  Differential privacy\\
 \hline
\end{tabular}
\caption{An informal summary of representative associative and causal nondiscrimination and privacy properties.
For each norm (row) and stance (column), the table provides an example of a property approaching that norm from that stance. We demonstrate that the core concepts of two rows are mathematically identical to one another
We also show that
columns represent a switch between causal and associative dependence. Table~\ref{tab:properties-summary3} represents each point in this grid in formal notation, exposing the structural correspondence between these properties.}
\label{tab:properties-summary-words}
\end{table}

As a result of the close correspondence between both stances of privacy and nondiscrimination
properties, we obtain a number of results `for free', which we describe in Section~\ref{sec:transfers}. In particular, the Dwork--Naor impossibility
result~\cite{dwork08jpc} translates to an impossibility of ensuring no disparate impact across
arbitrary subpopulations. Also, we point out research in both areas that can be repurposed
to solve the corresponding problem in the other area.

Of course, privacy and nondiscrimination are not the same.
While both deal with statistical associations and causal effects, they differ in which associations and causes are problematic.
Furthermore, our simple models of nondiscrimination and privacy abstract away the nuances of these social norms, such as exceptions to the general rules we represent.
However, these differences play little role in the mathematical analysis of or development of algorithms and verification techniques for the core properties capturing these norms.

There is, however, a key difference between nondiscrimination and privacy that is mathematically interesting: nondiscrimination focuses on adverse actions, whereas privacy sometimes deals with adverse actions but often deals with knowledge.
That is, while privacy and nondiscrimination share the dependence axis along which their properties differ,
privacy, unlike nondiscrimination, has a second axis.
This axis, the \emph{endpoint axis}, captures the difference between what has been called \emph{use privacy} (e.g.,~\cite{datta2017use}) and \emph{inferential privacy} (e.g.,~\cite{ghosh17itcs}).
Use privacy is similar to nondiscrimination in that it prohibits some actions from depending upon some sensitive fact.
Inferential privacy is more abstract in that it refers to the knowledge of an observer requiring that the observer maintains some degree of ignorance of some sensitive fact.

The two forms of privacy are related, but not equivalent.
An observer seeing an outcome that depends upon a sensitive fact may gain information about the sensitive fact, meaning that a lack of use privacy can imply a lack of inferential privacy.
Furthermore, an observer who knows the sensitive fact might use it to choose outputs inappropriately, meaning that a lack of inferential privacy can imply a lack of use privacy.
However, in both cases, \emph{can imply} does not mean \emph{does imply}.
An observer not knowing the dependence between the sensitive fact and the outputs gains no information and an observer may choose to not make use of its knowledge to inappropriately select outputs.

The mathematically interesting difference between use privacy and inferential privacy is use privacy should be measured in frequentist probabilities (or, more generally, physical probabilities), whereas inferential privacy should be measured in Bayesian probabilities.
Whereas the distinction between these two approaches to probability theory may appear to be merely a philosophical debate, we see here the practical distinctions between them.

In summary, we have three axes along which to navigate the space of properties:
(1) the norm: privacy or nondiscrimination,
(2) the notion of dependence: looking at either a change in association or in causation, and
(3) for privacy, the endpoint used: either measuring the change in use (frequentist probabilities) or in knowledge (Bayesian).

We will explore and make precise these differences and relations in this paper.
We first cover background and provide an overview of our results in Sections~\ref{sec:background} and~\ref{sec:overview}.
We then cover related work in Section~\ref{sec:related}.
In Section~\ref{sec:inf-flow}, we present a set of probabilistic definitions for formally stating causal and associative properties of systems, and prove relationships between their various forms. In Sections~\ref{sec:privacy} and~\ref{sec:nondiscrimination}, we instantiate
parameters in these definitions to obtain different notions of privacy and nondiscrimination respectively. Restating all properties in terms of this common substrate also allows us to transfer known theorems and methods about privacy to fairness, and vice versa, in Section~\ref{sec:transfers}.

\section{Background}
\label{sec:background}

\subsection{Probabilities, frequencies, and knowledge}

Probabilities are a useful tool to characterize a number of different concepts.
In this work, we distinguish between frequentist and Bayesian probabilities, or,
more generally, between physical and epistemic probabilities.

Frequentist probabilities represent the frequencies of
occurrence of events.
They are objective in that they measure a physical
property of the world.
They are a useful model for representing outcomes of
random coin flips or fractions of populations with a certain property.
Since they are often just called \emph{frequencies}, we denote frequentist probabilities using $\Fr$. Frequentist conditioning
restricts events to a smaller population. For example, if $O$ and $G$ are
random variables representing a certain hiring outcome and gender respectively,
then $\Fr[O = \text{`Hired'} \given G = \text{`Female'}]$ represents the frequency of women
that are hired out of the overall population, that is, the number of women hired divided by the size of the population restricted to just the females.
Note that in standard notation, the overall population is left implicit.

(Some authors use \emph{frequentist} and \emph{frequencies} in a strict sense limited to the case where the population size is infinite.  They might refer to such probabilities over finite populations as \emph{empirical frequencies} or \emph{physical probabilities}.  We note where this distinction may appear in our models of discrimination, but it can be safely ignored.)

Bayesian (or epistemic) probabilities represent
how certain a reasoning agent is about propositions being true.
They are subjective to the agent in that different agents may differ in the
probabilities they assign.
Since they are often just called credences, we denote Bayesian probabilities
using $\Cr$.  Bayesian conditioning represents a knowledge
update. For example $\Cr[O = \text{`Hired'} \given G = \text{`Female'}]$, represents the
certainty that a person was hired given that the agent knows that he or she was
female.
This decision may also depend upon the agent's background knowledge, which is considered fixed for the whole probabilistic analysis.
Thus, similar to the implicit overall population for frequencies, the background knowledge is typically left implicit.

In Table~\ref{tab:frequentist-bayesian}, we contrast a few probabilistic
notions and their interpretations in the frequentist and Bayesian views.
In situations in this paper where results apply to both forms, we denote probabilities
using $\Pr$.

\begin{table}
  \centering
\begin{tab}{lll}
 Concept  & Frequentist & Bayesian \\
  \midrule
 Probability      & Frequency  & Credence \\
 Probability space & Population & Background knowledge \\
 Conditioning     & Population restriction &  Knowledge revision  \\
 Association      & Co-occurrence & Evidence (information leakage)\\
\end{tab}
  \caption{Comparison of frequentist and Bayesian Concepts}
  \label{tab:frequentist-bayesian}
\end{table}

\subsection{Pearl's Account of Causation}
\label{sec:pearl}

We recount Pearl's theory of causation as background~\cite{pearl09book}.

\begin{dfn}
  A causal model is a triple $M = \langle U, V, F \rangle$ where $U$
  is a set of variables, called background variables (or exogenous), $V$ is set of variables, called
  endogenous, and $F$ is a set of functions ${f_1, \ldots,
  f_n}$, called structural equations, where each $f_i$ is a mapping that
  defines $V_i$ in terms of all other variables in $U\cup V$.
\end{dfn}

We assume, as Pearl normally does, that the equations are \emph{recursive}, that is, there is an ordering on $U \cup V$ such that all of $U$ comes before all of $V$, and that for each $f_i$, the variables in $U \cup V$ that it uses all come before $V_i$.
Under this assumption, given the values of the variables in $U$, one can compute the value for any $V_i$ in $V$ by computing the values of each variable in that order until reaching $V_i$.
We use $M.V_i(u)$ to denote the computed value.

An intervention $\doo(X{=}x)$ on an endogenous variable $X$, in a model $M$, replaces
the equation corresponding to $X$ with $x$, resulting in a new model $M_{X=x}$.

\begin{dfn}
A probabilistic causal model is a pair $\langle M, B\rangle$,
where $M$ is a causal model and $B$ is a background probability function defined over the domain of $U$, the background variables.
\end{dfn}
$B$ makes the the probabilities assigned to background variables explicit, whether that comes from an underlying population (frequentist) or background knowledge (Bayesian).
W%
For an assignment $u$ of values to all background variables $U$,
we will write $B(u)$ or $\Pr[U{=}u \given B]$, depending upon context.

For recursive models, the probabilities over background variables can be lifted to a
probability over an endogenous variable $Y$:
\[\Pr[Y{=}y \given B] = \sum_{\{u | M.Y(u){=}y\}} B(u).\]

The probability of counterfactual statements is defined by probabilities with respect to the model $M_{X{=}x}$.

Pearl prefers Bayesian probabilities~\cite{pearl09book}, but they can be interpreted either way.
From a Bayesian viewpoint, where these probabilities represent beliefs, the lifting of probabilities is only sensible under the assumption that the agent knows the structural equations.

We modify our notation slightly make this assumption explicit by showing the structural equations as explicit conditions in our probabilities.
Typically, our probabilities are of the form of $\Pr[\varphi \given SE, B]$ where the context comprises of
two parts:
the background probability distribution $B$ and
the structural equations $SE$.
$SE$ represents structural equations that define endogenous
variables. $SE$ contains exactly one structural equation for each endogenous
variable of the form $X = f(U_{i_1},\ldots, U_{i_k}, X_{j_1}, \ldots,
X_{j_k})$, where $U_{i_1},\ldots, U_{i_k}, X_{j_1}, \ldots, X_{j_k}$ are the
predecessors of $X$ in the structural model.

An intervention on an endogenous variable $X$ with value $x$ in the model $SE$
is denoted by $\isub{X}{x}SE$, which replaces the equation for $\inter{X}$ in $SE$ with
$\inter{X}{=}x$.

We assume that the background variables do not refer to or depend upon the structural equations.
In this case, since endogenous variables are defined by structural equations when the rest of a probability
expression only refers to background variables, structural equations can be substituted. In other
words, the choice of structural equations is independent to all other background variables.

\begin{assn}[SE-independence]\label{assn:se-independence}
  For each $\varphi$, consisting of only background variables, and background distribution $B$ for any sets of equations $SE$, $SE'$,
  \[\Pr[\varphi \given SE, B] = \Pr[\varphi \given SE', B]\]
\end{assn}

When using Bayesian probabilities, or credences, the background distribution
represents background knowledge of the agent doing the reasoning,
which we typically view as an adversary, and not some objective
population of outcomes.
In this case, Assumption~\ref{assn:se-independence} means that
adversary does not background knowledge about the background variables
that depends upon the structural equations.

We will often consider structural equation models representing the system of
the following form: $SE = \{\inter{X}\equals X, \inter{A}\equals A,
\inter{O}\equals \prog(\inter{X}, \inter{A})\}$, where $\inter{X}, \inter{A},
\inter{O}$ are endogenous variables, and $X$ and $A$ are background.
$\prog$ represents a system that is the subject of our privacy or nondiscrimination analysis.
$\inter{O}$ represents its output.
$X$ represents some sensitive attribute that, intuitively, $\prog$ should not use to compute $\inter{O}$.
$\inter{X}$ represents an input to the system representing $X$.
$A$ are other attributes that the system intuitively may use and $\inter{A}$ is the input to $\prog$ representing them.

Similar to Tschantz~et~al.~\cite{tschantz17arxiv}, we distinguish between the actual attribute $X$ and an input $\inter{X}$ representing it since we want to discuss how the system's output would change as its inputs change without considering all the other changes that would follow from changing the actual value of the sensitive attribute.
For example, a causal intervention to a person's race not only is conceptually difficult to comprehend~\cite{holland86jasa}, but would have a wide range of effects, such as likely resulting in a different spouse, address, and job, taking us far from the system $\prog$.
As a result, we don't allow intervening on the race itself variable
$X$ in our model, which is enforced by Pearl's theory disallowing
interventions on background variables.
On the other hand, since the endogenous variable $\inter{X}$ is
an input to well defined structural equations representing a system
$\prog$, it is possible to model the effects of intervening on the values
of this variable.
The same reasoning applies to $A$ and $\inter{A}$.

In our work, we assume that an adversary has complete knowledge of the model. While this
is a standard assumption in cryptography and privacy, interestingly, many leakage results
do not apply to weaker adversaries with partial knowledge about the model. In the weaker
adversary model, observing outputs can leak information about the model, which
can then leak information about secrets. However, this channel for information
leakage is not permitted in the case where the programs are completely known.

\section{Overview}
\label{sec:overview}

We categorize discrimination and privacy properties by the notion of
dependence they use: \emph{associative} or \emph{causal}
(see Table~\ref{tab:associative-causal}).
For each combination of norm and notion of dependence,
we provide a precise mathematical model of
a point in the space of fairness concepts.
For the privacy concepts, which have already been extensively studied by
the CS research community, we sometimes find pre-existing definitions
occupying one of these points.
In that case,
we restate the pre-existing definition in the standard form of
Tschantz~et~al.~\cite{tschantz17arxiv} to make its associative or
causal nature apparent.
We also note the kind of probability
employed. Statistical nondisclosure is a property about an adversary's
state of knowledge and is therefore represented as credences.
On the other hand,
differential privacy is a statement about the true distribution
over the outcomes of a program and is therefore represented by frequencies.
We also briefly consider an associative notion of privacy using frequencies
and a causal notion using credences.

The nondiscrimination properties come from complex legal tests that
include issues such as motivations and exceptions for unavoidable discrimination.
We will not model these tests in their completeness, but rather their
core conceptions of when discrimination (perhaps unavoidable and
legal) occurs.
Our models show that they too boil down to associative and causal
variants.

\begin{table}
  \centering
\begin{tab}{ll}
  Associative & Causal \\
  \hline
  Conditioning & Intervention \\
  Independence & Causal Irrelevance \\
  Association  & Influence\\
\end{tab}
  \caption{Comparison of causal and associative conceptions of dependence}
  \label{tab:associative-causal}
\end{table}

Our formalizations of all of these properties can be stated as a difference in outcomes across \emph{changes} in certain variables.
Each of the properties is about comparing some \emph{point of comparison} for two values that the a sensitive or secret attribute $S$ could have taken on.
The property demands that the point of comparison remains roughly the same whether $S$ took on $s$ or $s'$, meaning that the attribute $S$ plays at most a minor role in the determining the value of the point of comparison.

For example, recall disparate treatment, the U.S. legal standard for when people discriminated against \emph{because} of
having some protected attribute.
Intuitively, determining whether a woman has suffered from disparate treatment with respect to gender requires comparing two worlds: the actual world in which the women is female and a hypothetical world in which she were a male.
For each world, one computes the probability of the women experiencing an adverse outcome.
This calculation, done twice, is the point of comparison for disparate treatment.

For a more detailed example, recall that for a system $\prog$ to have $\epsilon$-differential privacy the following must hold
for all databases $d$ and $d'$ that differ by one person's data and all outputs $o$:
\begin{align}
 \Fr[\prog(d) = o] &\leq e^\epsilon  \Fr[\prog(d') = o]
\end{align}
(For simplicity, we only deal with discrete data in this paper, removing the need to consider sets of outputs.)
The point of comparison for differential privacy is the probability of $\prog$ having various outputs for various databases.
We will denote this as $\Fr[\prog(\changes{d}) = o]$ highlighting that the database changes from $d$ to $d'$ across the comparison by underlining it.
In addition to the point of comparison, differential privacy is also defined by
over what databases the comparison is made (those that differ by one entry),
over what outputs $o$ the comparison is made (all of them),
what the comparison demands (that $\cdot \leq e^\epsilon \cdot$ holds).
However, focusing on the point of comparison highlights the differences between properties we are interested in.

To make its causal form more apparent, following Tschantz~et~al.\@, we will rewrite differential privacy as
\begin{align}
 \Fr[\inter{O}\equals o \given \isub{\inter{X}}{x}SE] &\leq e^\epsilon \Fr[\inter{O}\equals o \given \isub{\inter{X}}{x'}SE]
\end{align}
using notation introduced in Section~\ref{sec:pearl}.
$\inter{O}$ is a random variable representing the output of the system $\prog$.
$\inter{X}$ represents the entry that changes between the two databases.
$x$ and $x'$ are the values by which the two databases differ, with a special value $\bot$ denoting that the entry is missing altogether.
(We are using the \emph{bounded} model of differential privacy, which differs slightly from the original definition, but not in a way material to the points we wish to make.)
Recall that $SE$ is the structural equations causally relating random variables.
In particular,
\begin{align}
\inter{X} &= X\\
\inter{A} &= A\\
\inter{O} &= \prog(\inter{X}, \inter{A})
\end{align}
where $X$ is a background variable corresponding to the actual value of sensitive attribute (the data point that changes),
$\inter{X}$ is the input to the system that represents this attribute,
$A$ is a background variable corresponding to the actual value of the other attributes (other data points),
$\inter{A}$ is the input to the system that represents these other attributes,
and $\inter{O}$ is the output.
$\Pr[\inter{O}\equals o \given \isub{\inter{X}}{x}SE]$ is the probability that the outcome variable $\inter{O}$ takes on the value $o$ given that a causal intervention set $\inter{X}$ to $x$.
This causal intervention differs from standard probabilistic conditioning in that it breaks correlations, preventing confounding, similar to how randomization does so in experiments.
Using this notation, the point of comparison is $\Pr[\inter{O}\equals o \given \isub{\inter{X}}{\changes{x}}SE)]$

A more extreme privacy property, a probabilistic version of noninterference~\cite{gm82security}, would demand equality:
\begin{align}
 \Fr[\inter{O}\equals o \given \isub{\inter{X}}{x}SE] &= \Fr[\inter{O}\equals o \given \isub{\inter{X}}{x'}SE]
\end{align}
For both differential privacy and noninterference, the point of comparison is the same, but the comparison relation differs from $\cdot \leq e^\epsilon \cdot$ and $\cdot = \cdot$.

Since our work does not discuss the tradeoffs between various comparison relationships, we will sometimes write just the point of comparison when discussing properties.
Using this shorthand, Table~\ref{tab:properties-summary3} lists and organizes six representative properties, or really representative clusters of properties.
In each case, we would have to also specify the comparison relation used and under what conditions the comparison is to be done to fully specify the property, but providing just point of comparison is sufficient to see the patterns that concern us in this work.

\begin{table}
\raggedright
\begin{tabular}{| p{2.25cm} | p{4.75cm} | p{4.75cm} |}
  \hline
                      &  Associative (conditioning)  &  Causal (intervening)   \\
  \hline
{\hyphenchar\font=-1 Discrimination} %
&  Disparate impact for $\inter{X}$ in population $B$:
\[ \Fr[\inter{O}\equals o \given \inter{X}\equals \changes{x}, SE, B] \]
& Disparate treatment on $\inter{X}$ for individuals in $B$:
\[ \Fr[\inter{O}\equals o \given \isub{X}{\changes{x}}SE, B] \]\\

 \hline

{\hyphenchar\font=-1 Privacy} %
&
Indirect use privacy~\S\ref{sec:indirect-use-privacy}:
\[ \Fr[\inter{O}\equals o \given \inter{X}\equals \changes{x}, SE, B] \]
Statistical nondisclosure of secret $\inter{X}$ and knowledge $B$~\cite{dalenius77statistik}:
\[ \Cr[\inter{X}\equals x \given \changes{\inter{O}\equals o} , SE, B] \]
\cellcolor{yellow!10}
& Noninterference~\cite{gm82security} and differential privacy~\cite{dwork06crypto} with secret $\inter{X}$:
\[ \Fr[\inter{O}\equals o \given \isub{X}{\changes{x}}SE, B] \]
\\
 \hline
\end{tabular}
\caption{Key clusters of properties.  For each cluster of properties, we show just one or more of its \emph{points of comparison}, the probabilities whose change in value should be minimized or bounded as the underlined term changes values.  In some cases, more than one pre-existing property may arise depending upon exactly how the comparison is made, exactly what is in secondary terms (e.g., $B$), and exactly under what circumstances the comparisons are made.  For example, the differences between noninterference and differential privacy include that noninterference requires exact equality for any two values of $\inter{X}$ whereas differential privacy requires approximate equality for only those values of $\inter{X}$ that differs in a single person's data.}
\label{tab:properties-summary3}
\end{table}

Table~\ref{tab:properties-summary3} reveals a similarity in both stances on the notion of dependence to use for privacy and nondiscrimination. The points of comparison are associative in the left column and causal in the right.
The properties in the left column minimize or bound associations.
The properties in the right column minimize or bound causal effects.

The quadrant for associative privacy is unlike the others in that we show two points of comparison belonging to this cluster of properties.
The first point of comparison, which we call \emph{indirect use privacy}, shows the correspondence between the four quadrants more clearly by differing from its neighbors in a minimal number of ways.
We could not find this point of comparison in prior work, although it is related to Pufferfish Privacy~\cite{kifer14database}, differing only by using frequencies instead of credences.
To help relate our four clusters to prior work, we also show statistical nondisclosure, one of the most well known associative privacy properties.
This property differs from the other quadrants' points of comparison by using credences and by flipping around the attributes that are measured and conditioned upon.
In Section~\ref{sec:privacy}, we take a closer look at the differences between associative privacy definitions.

Moving across the columns of the two rows of Table~\ref{tab:properties-summary3}, we see that for each stance on privacy, there exists an identical stance on nondiscrimination, and vice versa.
This correspondence shows the tight relationship between privacy and nondiscrimination, and the opportunity to reuse tools from one area of research in the other.

Despite being identical at the level of abstraction shown in Table~\ref{tab:properties-summary3}, privacy and nondiscrimination are, of course, not the same.
As mentioned, our models of discrimination and privacy only account for the core essence of some conceptions of these complex, multifaceted, and contested concepts.
Additionally, the secret or sensitive attribute $X$ will be instantiated differently for the two norms.
For example, race and gender are quintessential instantiations of $X$ for nondiscrimination, but it is harder to argue that they should be kept secret as an instantiation of $X$ for privacy.
Furthermore, the table contains statistical nondisclosure without a corresponding property for nondiscrimination.

\section{Related Work}
\label{sec:related}

While ours is the first comprehensive exploration of the correspondence
between causation and association, and privacy and nondiscrimination,
prior work has examined some of these connections in isolation.

\paragraph{Privacy and discrimination are similar.}
Implicitly making use of the similarity between privacy and nondiscrimination,
Dwork~et~al.\@ define a notion of fairness that requires
that similar people be treated similarly, and formalize this as a Lipschitz
continuity requirement\cite{dwork12itcs}.
They point out the relationship between this
notion of fairness as continuity and differential privacy: that differential
privacy can be viewed as a special case of continuity for the Hamming distance
metric on databases. We analyze this relationship through a causal lens and by
casting both nondiscrimination and differential privacy as restrictions on the
causal use of information.

In more detail,
mathematically, they use a metric $d$ on individuals that captures how
similar they are with respect to the classification task.
They also use a second metric $D$ over distributions of outcomes and
represent classifiers as a function $M$ from individuals to
distributions over outcomes.
They require that the $(D,d)$-Lipschitz property holds:
for all pairs of individuals $x$ and $y$, $D(M(x), M(y)) \leq d(x,y)$.
While they leave $d$ largely abstract since it is application specific,
they mostly focus on two possibilities for $D$:
statistical distance (total variation norm) and
the relative $\ell_\infty$ metric.
We focus on the second here since it is more similar to the other
definitions we have considered, and to differential privacy in particular
(their Section~2.3).
In our notation and treating it as a causal property, this would be the
requirement that for all individuals $x$ and $y$ and outputs $o$,
\begin{align}
\Fr[\inter{O}\equals o \given \isub{X}{x}SE, B]
\leq e^{d(x,y)} \Fr[\inter{O}\equals o \given \isub{X}{y}SE, B]
\end{align}
The point of comparison is
$\Fr[\inter{O}\equals o \given \isub{X}{\changes{x}}SE, B]$,
which is of the same form as differential privacy, although, here, $x$ ranges
over individuals instead of databases.

The authors object to group parity as insufficient for ensuring
fairness, although they do show conditions under which their condition
implies group parity (their Section~3), which are somewhat
similar in goal to our Theorems~\ref{thm:independence-intervention}
and~\ref{thm:epsilon-causal-epi-irr}, but rather different in form.
Like us, they use the connection between privacy and nondiscrimination
to transport results from privacy to nondiscrimination.
While they focus on algorithmic results (their Section~5)
rather different from our focus on concepts and impossibility results,
they do informally consider issues (their Section~3.1, Example~3, and
Section~6.3) similar to ones we discuss in our
Section~\ref{sec:nondiscrimination-interpretation}.

A difference between our work and theirs is that we focus on
prohibitions against using a certain attribute $X$ whereas they focus
on a requirement to use only a certain attribute implicitly defined by
the metric $d$.
While their approach is principled, it differs from current
antidiscrimination that places prohibitions on using protected
attributes (Section~\ref{sec:nondiscrimination}).
One could attempt to encode each approach into the other, but much of
the intuition would be lost even if successful.

\paragraph{Privacy and discrimination are different.}
Numerous works have instead looked at how privacy and nondiscrimination are not
the same.  Dwork and Mulligan write that approaching problems of discrimination
with the tools of privacy might just hide the
discrimination~\cite{dwork13slawr}.  Alan and Starr provide a concrete example:
the nondisclosure of criminal history appears to put pressure on employees to
discriminate by race as a proxy for the missing history~\cite{agan16tr}.
Strahilevitz also considers how having more information can reduce the desire
to discriminate~\cite{strahilevitz2008privacy}.

\paragraph{Privacy and causation.}
Tschantz~et~al.\@ are the first to provide a formal correspondence between
causation and information flow~\cite{tschantz15csf},
a result with implications for security, privacy, and nondiscrimination.
Others had
previously noted their relationship~\cite{mclean90sp,mowbray92csf,sewell00csf}.
This paper follows on the current authors' recent application of a %
similar correspondence to differential privacy~\cite{tschantz17arxiv}, which we
believe is the first explicit use of causal reasoning in privacy research.  This
viewpoint should not be confused with works that provide algorithms for causal
inference while providing differential privacy, such as
Kusner~et~al.'s~\cite{kusner16aistats}.

While we consider differential privacy
(at least when using the centralized model of data collection)
to be an instance of \emph{use privacy}, the
study of use privacy per se appears rather young.  A recent PCAST report has
called for more emphasis on the appropriate use of data, instead of banning
its collection~\cite{pcast14}.  Datta~et~al.\@ consider use
privacy and proxies~\cite{datta2017use}.

\paragraph{Privacy and association.}
The difference between associative privacy and differential (casual) privacy has received much discussion,
with some arguing for associative definitions (e.g.,~\cite{dalenius77statistik,kifer11sigmod,kifer12pods,kifer14database,he14sigmod,chen14vldbj,zhu15tifs,liu16ndss};
see~\cite{shannon49bell} for its antecedent in security) and others for differential ones (e.g.,~\cite{warner65asa,dwork06crypto,dwork06icalp,bassily13focs,kasiviswanathan14jpc,mcsherry16blog1,mcsherry16blog2}).

While associative properties have a more straightforward connection to providing inferential privacy, that is, limiting the inferences of an adversary from the data released,
differential privacy can also be viewed as limiting inferences.
In particular, Kasiviswanathan and Smith provide theorems showing that differential privacy implies a form of inferential privacy
that limits how much more an adversary can learn with an additional row in the database~\cite{kasiviswanathan14jpc}.

Others have looked implications of differential privacy about what the
adversary can learn from the released data in total given assumptions
about the adversary's knowledge~\cite{ghosh16inferential} or the
data~\cite{alivim11icalp}.
Others have looked at the relationship between differential privacy
and mutual information~\cite{cuff16ccs,mcsherry17blog1}.

\paragraph{Discrimination and causation.}
The importance of causal reasoning in nondiscrimination goes back to at least Pearl, who uses it to deal with how Simpson's paradox can make it unclear which if either of the two groups is discriminated against~\cite{pearl09book}.
More recently, Hardt~et~al.\@ re-examined this issue for a broad class of nondiscrimination definitions~\cite{hardt16arxiv,hardt16nips}.
Kilbertus~et~al.\@ propose a way of looking at this issue using causal reasoning about ``proxy'' and ``resolving'' variables that are either prohibited or allowed for use~\cite{kilbertus2017causal}.
Other causal notions of fairness have been put forward by
Kusner~et~al.~\cite{kusner2017counterfactual},
Bonchi~et~al.~\cite{bonchi2017exposing},
and Cowgill~et~al.~\cite{cowgill2017algorithmic}.

\paragraph{Discrimination and association.}
Associative notions of disparate impact have been the basis for a number
of mechanisms for enforcing nondiscrimination in statistical systems~\cite{calders10dmkd,kamishima2011, zemel13icml, feldman15kdd}. Recently a number
of richer associative notions of nondiscrimination have been proposed~\cite{hardt16nips,kleinberg17itcs,chouldechova16fatml} that take into account three variables, group membership, predicted outcome, and true outcome, in various combinations.
We believe our connection between privacy and nondiscrimination can be
extended from the simple associative notions to these more complex
ones, but leave it as future work.

\paragraph{Contrasting definitions in other ways.}
In this work, we focus on which distributions privacy and
nondiscrimination definitions compare, that is, on what we call the
\emph{point of comparison} found in the definition.
Mironov instead contrasts privacy definitions based upon how different
definitions use different methods of doing the
comparisons~\cite{mironov17csf}.

\section{Probabilistic notions of dependence}
\label{sec:inf-flow}

Before turning to privacy or nondiscrimination, we consider some general probabilistic notions that apply to both norms.
In particular, we look at properties that
compare probability distributions, frequentist or Bayesian, across two
worlds. In this section, we define a number of such properties, and prove
relationships between them.  In later sections, we instantiate
these properties to obtain privacy and nondiscrimination notions,
as well as theorems connecting these notions.

We first examine the simple case of properties for deterministic systems, and
 then, in
Section~\ref{sec:props:approx}, we examine definitions with approximate
guarantees for randomized systems.
We will typically state the properties as applied to our setting, modeled by $SE$,
even when the properties are more generally applicable to any system
of structural equations.
Table~\ref{tab:defs-summary} summarizes the properties we consider
and Figure~\ref{fig:relationships} summarizes the relationships
between them.

\begin{table*}
\centering
  \begin{tab}{llllcc}
Notion & Point of comparison                                      & Det. & AR \\ %
\midrule
Noninterference &    $\prog(\changes{h}, l)$                                        & \ref{dfn:ni} & \ref{dfn:ar-dc} \\ %
Associative independence &    $\Pr[\inter{O}\equals o \given X\equals \changes{x}, SE, B]$  & \ref{dfn:det-assoc}  & \ref{dfn:ar-assoc} \\ %
Associative independence &    $\Pr[X{=}x \given \underline{O{=}o}, SE, B]$                  &  & \ref{dfn:ar-assoc-X} \\
Causal irrelevance       &    $\Pr[\inter{O}\equals o \given \isub{X}{\changes{x}}SE, B]$   & \ref{dfn:det-causal}  & \ref{dfn:ar-causal}  \\ %
\end{tab}
  \caption{Summary of definitions. For each information flow property, we present
  a version for deterministic systems (Det.\@) and an approximate version for randomized systems (AR).
The point of comparison is the quantity computed twice, once for two different values, and compared to check whether they are equal to one another.
The check is for all pairs of values $s$ and $s'$ that can go in $\changes{s}$ (or $x$ and $x'$ for $\changes{x}$).}
  \label{tab:defs-summary}
\end{table*}

\begin{figure}
\centering
\begin{tikzcd}[math mode=false, column sep=5em, row sep=3em]%
                        & Det.                                                                 & AR \\
Associative independence &                                              & Def.~\ref{dfn:ar-assoc-X}  \\
Associative independence & Def.~\ref{dfn:det-assoc}\arrow[r]            & Def.~\ref{dfn:ar-assoc} \arrow[u, rightarrow, shift left=1ex, "Thm.~\ref{thm:assoc-to-assoc-X}"] \arrow[u, leftarrow, shift right=1ex, "Thm.~\ref{thm:assoc-X-to-assoc} $2\epsilon$"']\\
Causal irrelevance      & Def.~\ref{dfn:det-causal}\arrow[r]\arrow[u, leftrightarrow, "Thm.~\ref{thm:independence-intervention} $A \bot X \given B$"'] & Def.~\ref{dfn:ar-causal}\arrow[u, leftrightarrow, "Thm.~\ref{thm:epsilon-causal-epi-irr} $A \bot X \given B$"'] \\
Noninterference        & Def.~\ref{dfn:ni}\arrow[r]\arrow[u, leftrightarrow, "Thm.~\ref{thm:pr-ni}"']%
& Def.~\ref{dfn:ar-dc}\arrow[u, leftrightarrow, "Thm.~\ref{thm:diff-causal-irr}"']\\
\end{tikzcd}
\caption{Relationships between definitions.  Notable assumptions made by theorems and dilution of the privacy budget are shown as subscripts.  A deterministic form of the upper associative independence definition exists, but we do not consider it.}
\label{fig:relationships}
\end{figure}
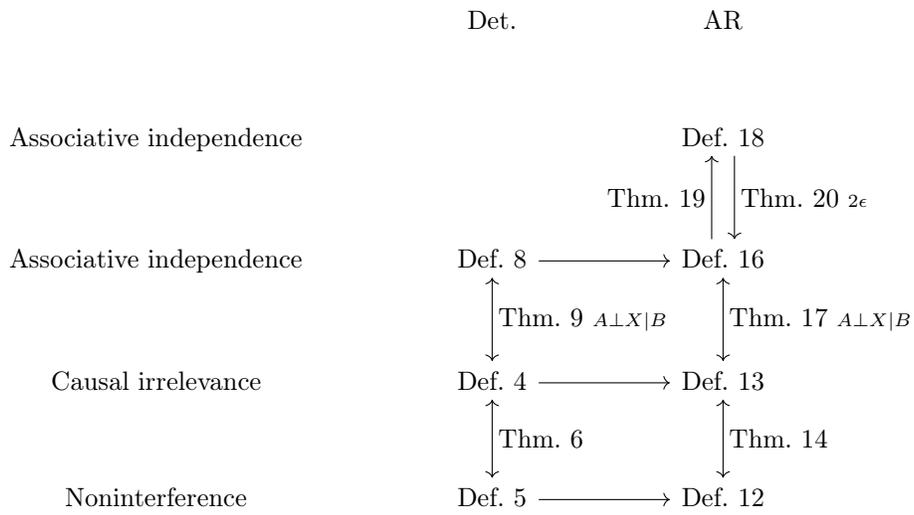

\subsection{The deterministic perfect case}
\label{sec:props:perfect}

The first property we consider is causal irrelevance.

\begin{dfn}[Pearl~\cite{pearl09book}]
  \label{dfn:det-causal}
  A system $SE = \{\inter{O}{=}\prog(\inter{X},\inter{A}), \inter{X}{=}X, \inter{A}{=}A\}$  has \emph{causal irrelevance} with respect to $\inter{X}$ to $\inter{O}$ for $B$ iff for all $x_1$, $x_2$, and $o$,
  \[\Pr[\inter{O}{=}o\given\isub{X}{x_1}SE, B] = \Pr[\inter{O}{=}o\given\isub{X}{x_2}SE, B]\]
\end{dfn}

For example, consider an experiment that randomly assigns a treatment of $x_1$ or $x_2$ to a
population modeled by $B$.
The causal irrelevance property states that the outcomes are
identical irrespective of the treatment.

For systems with control over all inputs, one approach for achieving causal
irrelevance is to enforce \emph{noninterference}.
Here we state a simplification of Goguen and Meseguer's
definition~\cite{gm82security}.

\begin{dfn}
  \label{dfn:ni}
  A function $\prog(\inter{X}, \inter{A})$ has noninterference for $\inter{X}$ iff for all $x_1$, $x_2$, and $a$.
  \[\prog(x_1, a) = \prog(x_2, a)\]
\end{dfn}

Noninterference yields \emph{causal irrelevance} for all backgrounds $B$.

\begin{thm}\label{thm:pr-ni}
  Consider a system $SE = \{\inter{O}{=}\prog(\inter{X},\inter{A}), \inter{X}{=}X, \inter{A}{=}A\}$. If $\prog$ has noninterference with respect to $\inter{X}$, then for all $B$ it has causal irrelevance with respect to $\inter{X}$.
\end{thm}

This theorem follows directly from the following lemma.

\begin{lem}
  \label{lemma:pr-ni}
  Consider a system $SE = \{\inter{O}{=}\prog(\inter{X},\inter{A}), \inter{X}{=}X, \inter{A}{=}A\}$. If $\prog$ has noninterference with respect to $\inter{X}$, then for all $B$ and $x$,
  \[\Pr[\inter{O}{=}o\given\isub{X}{x}SE, B] = \Pr[\inter{O}{=}o \given SE, B]\]
\end{lem}
\begin{proof}
  Assume $\prog(X, A)$ has noninterference with respect to $A$. Then,
  \begin{align*}
     & \Pr[\inter{O}{=}o\given\isub{X}{x}SE, B]\\
     =& \Pr[\inter{O}{=}o\given\inter{O}{=}\prog(\inter{X},\inter{A}), \inter{X}{=}x, \inter{A}{=}A, B]\tag{expanding $\isub{X}{x}SE$}\\
     =& \Pr[\inter{O}{=}o\given\inter{O}{=}\prog(x,\inter{A}), \inter{X}{=}x, \inter{A}{=}A, B]\tag{substitution}\\
     =& \Pr[\inter{O}{=}o\given\inter{O}{=}\prog(X,\inter{A}), \inter{X}{=}x, \inter{A}{=}A, B]\tag{by noninterference}\\
     =& \Pr[\prog(X,A){=}o\given\inter{O}{=}\prog(X,\inter{A}), \inter{X}{=}x, \inter{A}{=}A, B]\tag{substitution}\\
     =& \Pr[\prog(X,A){=}o\given\inter{O}{=}\prog(X,\inter{A}), \inter{X}{=}X, \inter{A}{=}A, B]\tag{SE-independence}\\
     =& \Pr[\inter{O}{=}o\given\inter{O}{=}\prog(\inter{X},\inter{A}), \inter{X}{=}X, \inter{A}{=}A, B]\tag{substitution}\\
     =& \Pr[\inter{O}{=}o\given SE, B]
  \end{align*}
\end{proof}

The second kind of property that we consider is associative independence.

\begin{dfn}
  \label{dfn:det-assoc}
  A system $SE = \{\inter{O}{=}\prog(\inter{X},\inter{A}), \inter{X}{=}X, \inter{A}{=}A\}$  has associative independence with respect to $B$ iff for each $x_1$, $x_2$, and $o$ such that $\Pr[X{=}x_1 \given SE, B] > 0$ and $\Pr[X{=}x_2 \given SE, B] > 0$,
  \[\Pr[\inter{O}{=}o\given X{=}x_1, SE, B] = \Pr[\inter{O}{=}o\given X{=}x_2, SE, B]\]
\end{dfn}

This is identical to stating that $\inter{O} \bot X ~\mid~SE, B$, that is, $\inter{O}$ and $X$ are conditionally independent with respect to $B, SE$.

Next, we show that if an input is independent of other inputs, then the causal irrelevance of that input is equivalent to the associative independence of that input.

\begin{thm}\label{thm:independence-intervention}
Consider a system $SE = \{\inter{O}{=}\prog(\inter{X},\inter{A}), \inter{X}{=}X, \inter{A}{=}A\}$. If $A \bot X \given B$ and for all $x$, $\Pr[X = x \given B] > 0$, then $SE$ has associative independence with respect to $\inter{X}$ for $B$ iff $SE$ has probabilistic causal irrelevance with respect to $X$ for $B$.
\end{thm}
\begin{proof}
  Follows directly from Lemma~\ref{lemma:independence-intervention}.
\end{proof}

\begin{lem}
  \label{lemma:independence-intervention}
  Consider a system $SE = \{\inter{O}{=}\prog(\inter{X},\inter{A}), \inter{X}{=}X, \inter{A}{=}A\}$. If $X \bot A \given B$, and $\Pr[X = x \given SE, B] > 0$,
  \[\Pr[\inter{O}\equals o \given X\equals x, SE, B] = \Pr[\inter{O}\equals o \given \isub{X}{x}SE, B]\]
\end{lem}
\begin{proof}
Assume that  $X \bot A \given B$.
Then,
\begin{align}
   \Pr[A\equals a \given B] &= \Pr[A\equals a \given B, X\equals x]\\
   \Pr[A\equals a \given SE, B] &= \Pr[A\equals a \given X\equals x, SE, B]\label{eqn:se-ind}
\end{align}
where \eqref{eqn:se-ind} follows from SE-independence.
Thus,
\begin{align*}
   & \Pr[\inter{O}\equals o \given X\equals x, SE, B]\\
   =& \Pr[\prog(x, A)\equals o \given X\equals x, SE, B] \tag{substituting for $\inter{X},\inter{O}$}\\
   =& \Pr[\prog(x, A)\equals o \given SE, B]\tag{from \eqref{eqn:se-ind} above}\\
   =& \Pr[\prog(x, A)\equals o \given \isub{X}{x}SE, B] \tag{SE-independence}\\
  =& \Pr[\inter{O}\equals o \given \isub{X}{x}SE, B]
\end{align*}
\end{proof}

\subsection{Making things approximate}
\label{sec:props:approx}

We now consider the case of systems with internal randomness.
We model such systems as $\{\inter{O}{=}\prog(\inter{X},\inter{A},R), \inter{X}{=}X, \inter{A}{=}A\}$, where
$R$ is a background variable that represents fresh randomness.

\begin{dfn}\label{dfn:fresh}
  A random variable $R$ with distribution $B_R$ is said to be \emph{fresh} for $B$ iff for all $\varphi$ that does not reference $R$ and all $r$, $\Fr[R\equals r \given B_R] = \Pr[R\equals r \given \varphi, B]$.
\end{dfn}
The left hand side represents the true frequency distribution of a random variable $R$ according to $B_R$.
The freshness condition can be instantiated for both frequentist and Bayesian probabilities. For
frequentist probabilities, this condition can be interpreted as $R$ being uncorrelated with any other background variables.
For Bayesian probabilities this condition can be interpreted as the agent knowing nothing about $R$ apart from its natural frequentist distribution $B_R$.
Since $R$ has the same distribution under $B_R$ and $B$ when it is fresh, we will not mention $B_R$ when $R$ is fresh for $B$.

Also, instead of requiring equalities, we allow the probabilities above to be
approximately equal, arriving at a version of differential
privacy for functions.

\begin{dfn}\label{dfn:ar-dc}
  A function $\prog(X, A, R)$ has $\epsilon$-noninterference for $X$ given a distribution $B$ over $R$ iff for all $x_1$, $x_2$, and $a$,
  \[\Fr[\prog(x_1, a, R){=}o \given B] \leq e^{\epsilon} \Fr[\prog(x_2, a, R){=}o \given B]\]
\end{dfn}
Note that the only role $B$ plays in Definition~\ref{dfn:ar-dc} is assigning a probability distribution to the randomization within $\prog$ provided by $R$.

\begin{dfn}
  \label{dfn:ar-causal}
  For a system $SE = \{\inter{O}{=}\prog(\inter{X},\inter{A},R), \inter{X}{=}X, \inter{A}{=}A\}$, $X$ has $\epsilon$-probabilistic causal irrelevance for $O$ with respect to $B$ iff for all $x_1$, $x_2$, and $o$,
  \[ \Pr[\inter{O}{=}o\given\isub{X}{x_1}SE, B] \leq e^\epsilon \Pr[\inter{O}{=}o\given\isub{X}{x_2}SE, B]\]
\end{dfn}

\begin{thm}
  \label{thm:diff-causal-irr}
  If $\prog(X, A, R)$ has $\epsilon$-noninterference, then for each $B$ such that $R$ is fresh, $X$ has $\epsilon$-probabilistic causal irrelevance for $O$ with respect to $B$, given $SE = \{\inter{O}{=}\prog(\inter{X},\inter{A},R), \inter{X}{=}X, \inter{A}{=}A\}$.
\end{thm}
\begin{proof} Assume $\prog$ has $\epsilon$-noninterference. Therefore,
  \begin{align*}
    &\Pr[\inter{O}{=}o\given\isub{X}{x_1}SE, B]\\
    &= \sum_a \Pr[\inter{O}{=}o\given A\equals a, \isub{X}{x_1}SE, B]\Pr[A\equals a\given \isub{X}{x_1}SE, B] \\
    &= \sum_a \Pr[\inter{O}{=}o\given A\equals a, \isub{X}{x_1}SE, B]\Pr[A\equals a\given B] \tag{exogeneity}\\
    &= \sum_a \Pr[\prog(x_1, a, R){=}o\given A\equals a, \isub{X}{x_1}SE, B]\Pr[A\equals a\given B] \\
    &= \sum_a \Fr[\prog(x_1, a, R){=}o \given A\equals a, B]\Pr[A\equals a\given B] \tag{SE-independence} \\
    &= \sum_a \Fr[\prog(x_1, a, R){=}o \given B]\Pr[A\equals a\given B] \tag{$R$ fresh} \\
    &\leq \sum_a e^{\epsilon} \Fr[\prog(x_2, a, R){=}o \given B]\Pr[A\equals a\given B] \tag{noninterference of $R$ and $B$} \\
    &= e^{\epsilon} \sum_a \Pr[\prog(x_2, a, R){=}o\given A\equals a, B]\Pr[A\equals a\given B] \tag{$R$ fresh}\\
    &= e^{\epsilon} \sum_a \Pr[\prog(x_2, a, R){=}o\given A\equals a, \isub{X}{x_1}SE, B]\Pr[A\equals a\given B] \tag{SE-independence}\\
    &= e^{\epsilon} \sum_a \Pr[\prog(x_2, a, R){=}o\given A\equals a, \isub{X}{x_2}SE, B]\Pr[A\equals a\given B] \\
    &= e^{\epsilon} \Pr[\inter{O}{=}o\given\isub{X}{x_2}SE, B]\\
  \end{align*}
\end{proof}

\begin{cor}
  \label{lemma:diff-pr-one-sided}
  If $\prog(X, A, R)$ has $\epsilon$-noninterference, then for all $B$ such that $R$ is fresh, for all $x$ and $o$,
  \[\Pr[\inter{O}{=}o\given SE, B] \leq e^\epsilon \Pr[\inter{O}{=}o\given\isub{X}{x}SE, B], \text{ and}\]
  \[\Pr[\inter{O}{=}o\given\isub{X}{x}SE, B] \leq e^\epsilon  \Pr[\inter{O}{=}o\given SE, B]\]
\end{cor}

\begin{proof}
  \begin{align*}
    &\Pr[\inter{O}{=}o \given SE, B] \\
  =& \sum_{x'} \Pr[\inter{O}{=}o \given X\equals x',  SE, B]\Pr[X\equals x'\given SE, B]\\
  =& \sum_{x'} \Pr[\inter{O}{=}o \given X\equals x',  \isub{X}{x'}SE, B]\Pr[X\equals x'\given \isub{X}{x'}SE, B]\tag{substituting for $X$}\\
  \leq& e^\epsilon \sum_x' \Pr[\inter{O}{=}o \given X\equals x', \isub{X}{x}SE,  B]\Pr[X\equals x'\given SE, B]\tag{Theorem~\ref{thm:diff-causal-irr}}\\
    =& e^\epsilon \Pr[\inter{O}{=}o \given\isub{X}{x}SE, B]\\
  \end{align*}
  The other direction follows similarly.
\end{proof}

\begin{dfn}
  \label{dfn:ar-assoc}
  A system $SE = \{\inter{O}{=}\prog(\inter{X},\inter{A},R), \inter{X}{=}X, \inter{A}{=}A\}$   has $\epsilon$-probabilistic associative independence on $\inter{O}$ with respect to $X$ for $B$ iff for each $x_1$, $x_2$ and $o$ such that  $\Pr[X{=}x_1\given SE, B] > 0$ and $\Pr[X{=}x_2 \given SE, B] > 0$,
  \[\Pr[\inter{O}{=}o\given X{=}x_1, SE, B] \leq e^\epsilon \Pr[\inter{O}{=}o\given X{=}x_2, SE, B]\]
\end{dfn}

$\epsilon$-probabilistic associative independence coincides with causal irrelevance when the sensitive input is independent of other inputs.

\begin{thm}
\label{thm:epsilon-causal-epi-irr}
  Consider a system $SE = \{\inter{O}{=}\prog(\inter{X},\inter{A},R), \inter{X}{=}X, \inter{A}{=}A\}$. If $X \bot A | B$, and $R$ is fresh, then, for $SE$, $X$ has $\epsilon$-probabilistic associative independence for $O$ with respect to $B$ iff $\inter{X}$ has $\epsilon$-probabilistic causal irrelevance for $O$ with respect to $B$.
\end{thm}

\begin{proof}
  This follows directly from Lemma~\ref{lemma:independence-intervention}.
\end{proof}

Finally, we show that the two forms of associative independence considered in Table~\ref{tab:defs-summary} are not that different.
Approximate associative notions of dependence may relax independence in a number of different ways.
The form of approximate associative independence in Definition~\ref{dfn:ar-assoc} is employed by the Pufferfish privacy formalism~\cite{kifer12pods,kifer14database}.
Statistical nondisclosure is a different form that for Bayesian probabilities compares posterior beliefs to prior beliefs~\cite{dalenius77statistik}.
An approximate form dropping the requirement of Bayesian probabilities follows:
\begin{dfn}
  \label{dfn:ar-assoc-X}
  A system $SE = \{\inter{O}{=}\prog(\inter{X},\inter{A},R), \inter{X}{=}X, \inter{A}{=}A\}$ has $\epsilon$-probabilistic associative independence on $X$ with respect to $O$ for $B$ iff for each $x_1$, $x_2$, and $o$ such that $\Pr[O{=}o\given SE, B] > 0$,
\begin{align}
 \Pr[X{=}x \given O{=}o, SE, B] &\leq e^{\epsilon} \Pr[X{=}x \given SE, B] &&\text{ and }\\
  \Pr[X{=}x \given SE, B] &\leq e^{\epsilon} \Pr[X{=}x \given O{=}o, SE, B]
\end{align}
\end{dfn}

We show here that the two formulations are very closely related. Since the structural equations $SE$, and background $B$ don't change throughout this section, we elide them from the statements in the rest of this section.

\begin{thm}\label{thm:assoc-to-assoc-X}
If for all $o$, $x_1$, and $x_2$, such that $\Pr[X{=}x_1\given SE, B] > 0$ and $\Pr[X{=}x_2\given SE, B] > 0$,
\[ \Pr[O{=}o \given X{=}x_1, SE, B] \leq e^{\epsilon} \Pr[O{=}o \given X{=}x_2, SE, B]\]
then, for all $o$ and $x$ such that $\Pr[O{=}o\given SE, B] > 0$,
\[\Pr[X{=}x \given O{=}o, SE, B] \leq e^{\epsilon} \Pr[X{=}x \given SE, B] \]
and
\[ \Pr[X{=}x \given SE, B] \leq e^{\epsilon} \Pr[X{=}x \given O{=}o, SE, B]\]
\end{thm}
\begin{proof}
  Assume for all $o$, $x_1$, and $x_2$ such that $\Pr[X{=}x_1\given SE, B] > 0$ and $\Pr[X{=}x_2\given SE, B] > 0$,
  \begin{align*}
\Pr[O{=}o \given X{=}x_1, SE, B] &\leq e^{\epsilon} \Pr[O{=}o \given X{=}x_2, SE, B]
\end{align*}
This implies that, for all $o$, $x_1$, and $x_2$ such that $\Pr[X{=}x_1\given SE, B] > 0$ and $\Pr[X{=}x_2\given SE, B] > 0$,
\begin{multline}
\Pr[O{=}o \given X{=}x_1, SE, B]\Pr[X{=}x_2 \given SE, B]
\\\leq e^{\epsilon} \Pr[O{=}o \given X{=}x_2, SE, B]\Pr[X{=}x_2 \given SE, B] \label{eqn:pufferfish-to-stat-nondis.withx2}
\end{multline}
Since \eqref{eqn:pufferfish-to-stat-nondis.withx2} holds for all $x_2$ such that $\Pr[X{=}x_2 \given SE, B] > 0$,
the following must also hold, for all $o$ and $x_1$ such that $\Pr[X{=}x_1 \given SE, B] > 0$:
{\footnotesize
\begin{align*}
\sum_{x_2 \in \mc{X}'}\Pr[O{=}o \given X{=}x_1, SE, B]\Pr[X{=}x_2 \given SE, B] &\leq \sum_{x_2 \in \mc{X}'} e^{\epsilon} \Pr[O{=}o \given X{=}x_2, SE, B]\Pr[X{=}x_2 \given SE, B]\\
\Pr[O{=}o \given X{=}x_1, SE, B]\sum_{x_2 \in \mc{X}'}\Pr[X{=}x_2 \given SE, B] &\leq e^{\epsilon} \sum_{x_2 \in \mc{X}'}\Pr[O{=}o \given X{=}x_2, SE, B]\Pr[X{=}x_2 \given SE, B]\\
\Pr[O{=}o \given X{=}x_1, SE, B] * 1 &\leq e^{\epsilon} \sum_{x_2 \in \mc{X}}\Pr[O{=}o \land X{=}x_2 \given SE, B]\\
\Pr[O{=}o \given X{=}x_1, SE, B] &\leq e^{\epsilon} \Pr[O{=}o \given SE, B]\\
\Pr[O{=}o \given X{=}x_1, SE, B]\Pr[X{=}x_1 \given SE, B] &\leq e^{\epsilon} \Pr[O{=}o \given SE, B]\Pr[X{=}x_1 \given SE, B]
\end{align*}
}
where $\mc{X}_2$ is the range of $X$ and $\mc{X}'$ is the support of $X$ given $SE$ and $B$:
$\set{x \in \mc{X}}{\Pr[X{=}x \given SE, B] > 0}$.
From this it follows that
for all $o$ and $x_1$ such that
$\Pr[O{=}o \given SE, B] > 0$ and
$\Pr[X{=}x_1 \given SE, B] > 0$:
\begin{align*}
\frac{\Pr[O{=}o \given X{=}x_1, SE, B]\Pr[X{=}x_1 \given SE, B]}{\Pr[O{=}o \given SE, B]} &\leq e^{\epsilon} \Pr[X{=}x_1 \given SE, B]\\
\Pr[X=x_1\given O{=}o, SE, B] &\leq e^{\epsilon} \Pr[X{=}x_1 \given SE, B]
\end{align*}
In the case where $\Pr[X{=}x_1 \given SE, B] = 0$, the final inequality also holds since both sides are $0$.

Similarly, we can show that
\[\Pr[X=x_1 \given SE, B] \leq e^{\epsilon} \Pr[X{=}x_1\given O{=}o, SE, B]\]
\end{proof}

\begin{thm}\label{thm:assoc-X-to-assoc}
If for all $o$ and $x$ such that $\Pr[O{=}o\given SE, B] > 0$,
\[\Pr[X{=}x \given O{=}o, SE, B] \leq e^{\epsilon} \Pr[X{=}x \given SE, B]\]
and
\[\Pr[X{=}x \given SE, B] \leq e^{\epsilon} \Pr[X{=}x \given O{=}o, SE, B]\]
then, for all $o$, $x_1$, and $x_2$, such that $\Pr[X{=}x_1\given SE, B] > 0$ and $\Pr[X{=}x_2\given SE, B] > 0$,
\[\Pr[O{=}o \given X{=}x_1, SE, B] \leq e^{2\epsilon} \Pr[O{=}o \given X{=}x_2, SE, B]\]
\end{thm}
\begin{proof}
  Assume for all $o$ and $x$ such that $\Pr[O{=}o\given SE, B] > 0$,
\begin{align*}
\Pr[X{=}x \given O{=}o, SE, B] &\leq e^{\epsilon} \Pr[X{=}x \given SE, B]
\end{align*}
From this and Bayes's Rule, it follows that
for all $o$ and $x$ such that
$\Pr[O{=}o\given SE, B] > 0$ and
$\Pr[X{=}x\given SE, B] > 0$,
\begin{align*}
\frac{\Pr[O{=}o \given X{=}x, SE, B]\Pr[X{=}x \given SE, B]}{\Pr[O{=}o \given SE, B]} &\leq e^{\epsilon} \Pr[X{=}x \given SE, B]\\
\Pr[O{=}o \given X{=}x, SE, B] &\leq e^{\epsilon} \Pr[O{=}o \given SE, B]\\
\end{align*}
The last line also holds in the case where $\Pr[O{=}o\given SE, B] > 0$ since both sides will be zero.

Similarly, we can show that for all $o$ and $x$ such that $\Pr[X{=}x \given SE, B] > 0$,
\[ \Pr[O=o \given SE, B] \leq e^{\epsilon} \Pr[O{=}o\given X{=}x, SE, B] \]

Putting the two together, we get
for all $o$, $x_1$, and $x_2$ such that
$\Pr[X_1{=}x_1\given SE, B] > 0$ and
$\Pr[X_2{=}x_2\given SE, B] > 0$,
\[ \Pr[O{=}o \given X{=}x_1, SE, B] \leq e^{\epsilon} \Pr[O{=}o \given SE, B] \leq e^{2\epsilon} \Pr[O{=}o \given X{=}x_2, SE, B] \]
\end{proof}

\section{Privacy}
\label{sec:privacy}

We now instantiate the definitions above to obtain various definitions of
privacy reported in prior work.
We will organize this section around the two types of privacy definitions found in Table~\ref{tab:properties-summary3}.
For each type, we will first discuss its general nature and then relate it to the most well known formal notion of privacy, differential privacy.

\subsection{Direct Use Privacy: Frequentist Causal Irrelevance}
\label{sec:direct-use-privacy}

The first kind of privacy requirements we examine is \emph{direct use
privacy}, which requires that some data be not used to produce some output.
One way of
formalizing non-use is Pearl's notion of causal irrelevance, which holds if
the outcomes over a population represented by $B$ are identical when the
sensitive input is intervened on.
Since in this section we will be using more than one type of background
distribution $B$, we will denote ones modeling populations as $B_P$.
This view is the basis for Tschantz et al.\@'s prior work on
inferring information use by detecting causation~\cite{tschantz15csf}.
Their work started with \emph{noninterference}, the classic security
property formalizing information flow~\cite{gm82security}, which is
mathematically equivalent to direct use privacy despite the
difference in application.
It showed that noninterference is equivalent to Pearl's notion of
causation.
That is,
absolute direct use privacy holds iff for all $o$, $x_1$, $x_2$, and $B_P$,
\begin{align}
 \Fr[\inter{O}\equals o \given \isub{X}{x_1}SE, B_P]
 &= \Fr[\inter{O}\equals o \given \isub{X}{x_2}SE, B_P]
   \label{eqn:direct-use-privacy}
 \end{align}
where
$SE = \{\inter{O}{=}\prog(\inter{X},\inter{A}), \inter{X}{=}X, \inter{A}{=}A\}$
as usual.
Differential privacy can be understood as a relaxation of this property, which
we'll cover in detail below.

\paragraph{Differential Privacy}

Differential privacy is a relaxed form of direct use privacy.
It can be expressed as either $\epsilon$-noninterference or $\epsilon$-causal irrelevance for database rows.
For the setting of database privacy, where a function operates on a database comprising a set of
rows, differential privacy requires that the distribution over outcomes is not significantly affected
by the value of a single row in the database.

\begin{dfn}
  \label{dfn:dp}
  A system $\prog(d, R)$, operating over a data set $d$ has \emph{$\epsilon$-differential privacy} iff for all rows $i$, databases $d$, and row values $x$, $x'$, and outputs $o$
  \[\Fr[\prog(d_{-i}x, R) = o] \leq e^\epsilon \Fr[\prog(d_{-i}x', R)]\]
\end{dfn}

In the definition above, for a database $d$, $d_{-i}$ refers to the remainder of the database with
the row indexed by $i$ left out, and $d_{-i}x$ refers to the database $d$ with the value $x$ in row $i$.
We allow a special value $\bot$ for $x$ or $x'$ where $d_{-i}\bot$ denotes the database without replacing the $i$th entry.
(This is the so-call \emph{bounded} formulation of differential privacy.)
In the notation we used before, $d_{-i}$ can be thought of as the value of $\inter{A}$ and $x$ as the value of $\inter{X}$ with the split of the data points into $\inter{A}$ and $\inter{X}$ varying with the value of $i$.
$R$, as before, represents the randomness used by $s$.
That is, differential privacy is $\epsilon$-noninterference with respect to each
row in the database.
It means that intervening on any row does not affect the distribution
over outcomes significantly.
\begin{prp}
  A function $\prog(D, R)$, where $D = \langle D_1, \cdots, D_k\rangle$, has $\epsilon$-differential privacy iff for each $i$, $\prog(D, R)$ has $\epsilon$-noninterference with respect to $D_i$.
\end{prp}
This equivalence follows directly from the definitions.
Moreover, we can express differential privacy as causal irrelevance:
\begin{prp}
  Consider system defined by the structural equations
  $SE = \{
  \inter{O}{=}\prog(\inter{D},R),
  \inter{A}{=}A,
  \inter{D}_1{=}D_1, \ldots, \inter{D}_k{=}D_k\}$.
  The function $\prog$ has $\epsilon$-differential privacy iff for all $i$ and all $x_1$ and $x_2$, and for all $B_P$,
\[\Fr[\prog(\inter{D}, R) \given \preisub{\inter{D}_i}{x_1}SE, B_P] = e^\epsilon \Fr[\prog(\inter{D}, R) \given \preisub{\inter{D}_i}{x_2}SE, B_P] \]
where we use $\inter{D}$ as short for $\langle\inter{D}_1,\ldots,\inter{D}_k\rangle$.
\end{prp}
This equivalence is the corollary of Theorem~\ref{thm:diff-causal-irr}.
Tschantz~et~al.\@ provide details~\cite{tschantz17arxiv}.

\subsection{Associative Inferential Privacy: Bayesian Associative Independence}

Associative inferential privacy tries to capture the change in the beliefs of an adversary
after observing the outcome of a program. This notion is well captured
by measuring the difference in the Bayesian distributions representing the
adversary's beliefs with and without conditioning on the outcome.
That is, if $\varphi$ is the proposition of interest that the
adversary is attempting to learn about, we compare
$\Cr[\varphi \given O\equals o, B_K]$ to $\Cr[\varphi \given B_K]$
where $B_K$ is a background distribution representing the adversary's
background knowledge about background variables.
Given that the definition we consider in this section does not use
causal interventions, $B_K$ can more generally refer to any background
knowledge.
In the absolute case, we require that
\begin{align}
\Cr[\varphi \given O\equals o, B_K] &= \Cr[\varphi \given B_K]
\end{align}
This concept goes back to at least Dalenius in 1977 as \emph{statistical nondisclosure}~\cite{dalenius77statistik}.
The definition looks more familiar if we consider the special case where $\varphi$ is $X = x$:
\begin{align}
\Cr[X{=}x \given O\equals o, B_K] &= \Cr[X{=}x \given B_K] \label{eqn:stat-nondisc-x}
\end{align}

This requirement is closely related to a similar special case Pufferfish privacy~\cite{kifer12pods,kifer14database}:
\begin{align}
\Cr[\inter{O}\equals o \given \inter{X}\equals x_1, SE, B_K]
&= \Cr[\inter{O}\equals o \given \inter{X}\equals x_2, SE, B_K]
\label{eqn:pufferfish}
\end{align}
Section~\ref{sec:props:approx} shows approximate forms of these requirements that, similar to differential privacy allows them to differ by a factors of $e^\epsilon$.
The section also shows that they are nearly equivalent, with
the relaxation of \eqref{eqn:pufferfish} implying the relaxation of \eqref{eqn:stat-nondisc-x}
and
the relaxation of \eqref{eqn:stat-nondisc-x} implying the relaxation
of \eqref{eqn:pufferfish} with the privacy budget diluted from
$\epsilon$ to $2\epsilon$.

Other approaches to measuring an associative difference
include measuring a notion of accuracy~\cite{clarkson05csf} or a difference
in beliefs over the runs of a system~\cite{halpern08tissec}.

Such associative inferential requirements are hard to meet.
Typically, we would not know the background knowledge of the adversary.
In the worst case, we might need to consider any background knowledge as possible.
Dwork and Naor have shown that it is impossible, in
general, to both satisfy this requirement for all background knowledge
sets $B_K$ and release an output $O$ providing utility~\cite{dwork08jpc}.
In Section~\ref{sec:imposs}, we review and re-present Dwork and Naor's
proof.

\paragraph{Differential Privacy implies Associative Inferential Privacy for Independent Data Points}

One way to get around the aforementioned impossibility result is to
attain the property for a restricted set of background knowledge.
If we require that the data points be independent under the adversary's
background knowledge, then differential privacy implies associative inferential privacy.
We can now use the results proved in Section~\ref{sec:props:approx} for differential privacy.

The implication states that for background knowledge such that each row
is independent, conditioning on the value of a row does not change the beliefs
about the outcomes. The requirement for independence to obtain such an
epistemic guarantee has been pointed out in prior work~\cite{kifer12pods,kifer14database}, and
follows directly from Theorem~\ref{thm:diff-causal-irr} and
Theorem~\ref{thm:epsilon-causal-epi-irr}.

\begin{prp}
\label{prp:dp-implies-assoc-inf-priv}
  Consider a system $SE = \{\inter{O}{=}\prog(\inter{D},R), \inter{D}{=}D,
  \inter{A}{=}A\}$ and $\inter{D} = \langle\inter{D}_1,
  \cdots\inter{D}_k\rangle$, and $D = \langle D_1, \cdots, D_k\rangle$. If
  $D_i\bot D_{-i} | B_K$ for all $i$, and $R$ is fresh, then for all $i$, $d_i$, $d_i'$, and $o$,
  \[\Cr[\inter{O}\equals o\given \inter{D}_i=d_i, SE, B_K] \leq e^{\epsilon} \Cr[\inter{O}\equals o\given \inter{D}_i=d_i', SE, B_K]\]
  if $\prog$ is $\epsilon$-differentially private.
\end{prp}

Another reasonable restriction on background knowledge is to tie the background knowledge to an objective population frequency distribution: the distribution that captures the associations known about a population or set of populations.
This approach has its own challenges that we discuss below.

\subsection{Indirect Use Privacy: Frequentist Associative Independence}
\label{sec:indirect-use-privacy}

Above, we considered a frequentist causal property and a Bayesian
associative property.
Here, we consider the possibility of a frequentist associative notion
of privacy.
Let \emph{indirect use privacy} be the requirement that the
output is not statistically associated with the sensitive attribute $X$.
For a population $B_P$, absolute indirect use privacy holds iff for all $o$, $x_1$, and $x_2$,
\[ \Fr[\inter{O}\equals o \given \inter{X}\equals x_1, SE, B_P]
 = \Fr[\inter{O}\equals o \given \inter{X}\equals x_2, SE, B_P] \]
where
$SE = \{\inter{O}{=}\prog(\inter{X},\inter{A}), \inter{X}{=}X, \inter{A}{=}A\}$
as usual.
This property can be understood as a more objective version of associative inferential privacy.
That notion of privacy is subjective in the sense that it depends upon the state of knowledge of an adversary.
This definition replaces that background knowledge $B_K$ with the actual population $B_P$ as it switches from credences to frequencies, making it more objective by instead referring to the state of the outside world.
Furthermore, this requirement will imply associative inferential privacy for adversaries whose background knowledge is limited to knowing the population.
\mct{Ideally, a proof would be provided.}

Indirect use privacy prevents the release of any output that is associated with the sensitive attribute, even if that release is not caused by the sensitive attribute.
For example, suppose the sensitive attribute $X$ is whether someone has cancer.
Indirect use privacy would prevent the release of whether someone smokes, which is typically public knowledge, as the output $O$ since it is associated with getting cancer.

Generally, the population $B_P$ might not be completely known to whoever designs the system.
Guaranteeing indirect use privacy, in this case, requires meeting the definition for all populations $B_P$ that the designer believes to be possible.
In the extreme case, this would require protecting the information for all populations $B_P$, leading to the same impossibility result as for associative inferential privacy.

Perhaps due to how strong this requirement is,
we do not know of any work that attempts to achieve this requirement for a comprehensive class of populations $B_P$.
However, some prior works can be viewed as can be viewed as approximations of this requirement.
We will discuss two such approximations, one for privacy in general and one related to differential privacy in particular.

\paragraph{Proxies}
Prior work has prohibited the direct use of \emph{proxies},
which are variables statistically associated with the sensitive
attribute $X$ (e.g.,~\cite{datta17ccs}), for example ZIP Codes and race.
Directly using a variable that is associated with $X$ is
required to violate indirect use privacy assuming that the system
cannot predict the value of $X$ \emph{a priori}.
However, proxies may be used as inputs without introducing associations between the outcomes
and sensitive attributes, such as the use of ZIP Code to target  geographic locations
that may have heterogeneous racial demographics. Identifying
the use of proxies in a system exposes points at which normative
judgments can be made about the acceptable use of proxies.

\paragraph{Inferential Guarantees from Differential Privacy}

Ghosh and Kleinberg reason about what an adversary could learn from a
differentially private data release when the adversary's knowledge is
characterized by the correlations in the underlying
population~\cite{ghosh17itcs}.
In our notation, this corresponds to assuming that $B_K$ is determined
by $B_P$.
They show that the amount of associative inferential privacy lost is
bounded when the correlations in $B_P$ are bounded in particular ways.

\paragraph{Definitions Going Beyond Differential Privacy.}

Recall that Proposition~\ref{prp:dp-implies-assoc-inf-priv}
showing that differential privacy
implies a form of associative inferential privacy under the
pre-condition that the data points are independent under the
adversary's background knowledge $B_K$.
That pre-condition is not optional:
numerous works have noted that the output of a differentially private
function can allow an adversary to draw an inference about a single
data point due to output still being strongly associated (more so than
by a factor of $e^\epsilon$) with the data point
(e.g.,~\cite{kifer11sigmod,kifer12pods,kifer14database,he14sigmod,chen14vldbj,zhu15tifs,liu16ndss}).
Some of these authors have responded by proposing definitions that are
similar to differential privacy but requiring that a larger set of
databases are treated as neighboring~\cite{kifer12pods,kifer14database,chen14vldbj,zhu15tifs,liu16ndss}.
This means that the algorithm will have to produce nearly identical
(within a factor of $e^\epsilon$) distributions over outputs for
more pairs of databases.
Intuitively, these additional pairs of databases are ones that differ
not in the data point $\inter{D}_i$ currently under consideration, but
rather other data points $\inter{D}_j$ that are closely associated
with $\inter{D}_i$.
By requiring the output to not depend much upon these associates, such
definitions are similar to proxy prohibitions where the correlated
data points $D_j$ are treated as the prohibited proxies.
By prohibiting their use, these definitions attempt to avoid a strong
association Between the output and a single data point.
Between this motivation and being implemented as a restriction over
outputs, these definitions are similar to indirect use privacy in the
same sense as proxy prohibitions are similar.

\subsection{Causal Inferential Privacy}

We have seen Bayesian associative notions of privacy and frequentist causal ones.
This raises the question of what Bayesian causal ones exist.
Intuitively, such causal epistemic properties require that
intervening on the sensitive attribute does not change significantly an
observer's beliefs about any proposition $\varphi$ about the background.

A trivial way of getting such a requirement is to replace the
frequencies in the causal notion of privacy found in
Section~\ref{sec:direct-use-privacy} \eqref{eqn:direct-use-privacy}
with Bayesian probabilities.
Given Pearl's preference for using Bayesian
probabilities~\cite{pearl09book}, such a replace would bring our causal
notion into closer correspondence to his model of causation.

However, doing such a wholesale replace ignores the distinction
between probabilities measuring properties of a population $B_P$ and
properties of the adversary's background knowledge $B_K$.
A more nuanced approach would use both forms of probabilities,
combining frequencies, credences, causal interventions.
We leave this effort to future work.

\paragraph{Semantic (Differential) Privacy}

Kasiviswanathan and Smith provide a ``semantic'' version of
differential privacy, which they call
\emph{semantic  privacy}~\cite{kasiviswanathan14jpc}.
It requires that the probability that the adversary assigns to the all
input data points does not change much whether an individual $i$ submits
data or not.
While they did not express their definition in terms of causation, we
conjuncture that this definition could be expressed in such terms,
using a combination of frequencies and credences.
Indeed, while they do not formally distinguish between frequencies and credences,
their notation suggests such a distinction as they classify some probabilities as coming
from the algorithm (their $\Pr$) and
others from the adversary's beliefs (their $b$).
Kasiviswanathan and Smith prove that differential privacy and
semantic (differential) privacy are closely
related~\cite[Thm.~2.2]{kasiviswanathan14jpc},
and such a result should also apply to our causal view.

\section{Nondiscrimination}
\label{sec:nondiscrimination}

Similar properties to the ones discussed above appear as
non\-dis\-crim\-i\-nat\-ion properties.
For example, \emph{group parity}, a nondiscrimination property
that the frequency of individuals hired from a protected group be similar to
the frequency of individuals hired from the rest of the population,
is an associative notion.
Furthermore,
the use of a protected attribute such as race or gender can be viewed as
a causal property, and is an important part of the concept of disparate
treatment.

Much as differential privacy served as a point of reference for
understanding the forms of privacy properties, we use the U.S.\@ legal
notions of disparate treatment and disparate impact here.
However, whereas, differential privacy fit squarely under one of our
forms of privacy, the situation is more complex here due to the
complexity of the law.

\subsection{Direct Nondiscrimination: Frequentist Causal Irrelevance}

The first kind of nondiscrimination requirements we examine is
\emph{direct nondiscrimination}, which requires that some protected
attribute not be used to produce some output.
As with direct use privacy, we formalize non-use as Pearl's notion of causal irrelevance.
That is,
absolute direct nondiscrimination holds iff for all $o$, $x_1$, $x_2$, and $B_P$,
\[ \Fr[\inter{O}\equals o \given \isub{X}{x_1}SE, B_P]
 = \Fr[\inter{O}\equals o \given \isub{X}{x_2}SE, B_P]
 \]
where
$SE = \{\inter{O}{=}\prog(\inter{X},\inter{A}), \inter{X}{=}X, \inter{A}{=}A\}$
as usual.
This property is identical absolute direct use privacy.
However, for privacy, the variable $X$ would intuitively be some attribute kept
private, perhaps a medical diagnosis, whereas
for nondiscrimination, the variable $X$ would intuitively be some
protected attribute.
For example, $X$ could be gender or race, which are typically
publicly known.

\paragraph{Disparate Treatment}

The above formulation is related to the legal notion of \emph{disparate impact}.
In U.S. law, Section~703 of Title~VII of the Civil Rights Act of 1964 (codified as 42~U.S.C.~§2000e-2) contains
\begin{quote}
(a) It shall be an unlawful employment practice for an employer --
\begin{itemize}
\item[(1)] to fail or refuse to hire or to discharge any individual, or otherwise to discriminate against any individual with respect to his compensation, terms, conditions, or privileges of employment, because of such individual's race, color, religion, sex, or national origin; or [\ldots]
\end{itemize}
\end{quote}
Disparate treatment is one way in which this provision could be violated.
The other is \emph{disparate impact}, which we turn to next.
The U.S.~Supreme Court has described \emph{disparate treatment} as follows~\cite[Footnote~15]{stewart77scotus}:
\begin{quote}
``Disparate treatment'' such as is alleged in the present case is the most easily understood type of discrimination. The employer simply treats some people less favorably than others because of their race, color, religion, sex, or national origin. Proof of discriminatory motive is critical, although it can in some situations be inferred from the mere fact of differences in treatment.
\end{quote}
This definition leaves some ambiguity about whether unintentional but
causal discrimination counts as disparate treatment.
For example, suppose an employer harbors a subconscious basis against
women that causes him to not hire them despite being unaware of it.
Arguably, sex would be the cause of less favorable treatment, but
without intent.

We suspect that such ambiguity has not posed problems in legal cases for two reasons.
First, proving, outside of a laboratory setting,
the existence of such a subconscious bias and its role in decision
making is so difficult that such cases are unlikely to arise due to a lack of evidence.
Second, these difficulties do not apply to showing disparate impact,
providing plaintiffs with an easier alternative course of action in
such cases, which means that few would feel tempted to push the
envelope on what counts as disparate treatment.

Turning to automated systems, this ambiguity may need to be
addressed.
Some have argued that machines can, in a sense, have intent
(e.g.,~\cite{dennett87book}), but, arguably, they cannot.
With this in mind, perhaps, the notion of disparate treatment is best
limited to humans.
However, we are interested in the case where a system uses protected
information about a person to select an outcome, that is, systems in
which the protected output causes the outcome.
Such cases are violations of direct nondiscrimination and
would be illegal under the quotes above given a causal
interpretation of the word \emph{because} found in them both.
While this interpretation loses some of the nuances of intent, we
believe it maintains the essence of the stance that gives raise to it:
no one should be treated poorly as an effect their protected attributes.

Under this view, at a high level, showing disparate treatment requires
establishing direct discrimination.
Experimental methods such as randomized controlled experiments and
situation testing over populations can establish such causation.

In practice, disparate treatment cases look little like scientific studies
establishing causation for two reasons.
First, running such experiments over actual workplaces is typically
impossible except in limited circumstances (see, e.g.,~\cite{romei14survey}).
Second, winning a disparate treatment case is more complex than
just showing direct discrimination.
For one, nondiscrimination laws only cover certain entities making
certain types of decisions.
Further complicating matters, exceptions exist.
For example, a gym will typically be able to defend hiring only
females to attend its women's locker room as a
\emph{bona fide occupational qualification}, which permits an
exception to the prohibitions against employment discrimination.
Furthermore, showing direct discrimination yields little
reward in \emph{mixed-motives} cases.
For example, the defendant may show that
even if it had not discriminated against the plaintiff, the plaintiff
would still not have been hired for some other legal reason, such as
lacking a required quantification.
In such cases, under U.S.~law, the defendant does not receive damages
(42~U.S.C.~\S{}2000e-5(g)(2)(B)).

For these reasons, legal formulations and case law surrounding disparate
treatment centers around dis\-crim\-i\-nat\-ion against individuals~\cite[Chapter
10.1]{united2008ninth}, where a plaintiff claims that a protected attribute was
a sole reason or motivating factor for some adverse action against them
and had it not been for the protected attribute the adverse action
would not have happened.
(Meanwhile, the defendant typically claims that this is not the case and
offers other lawful reasons that supposedly motivated the adverse action.)
Showing that the adverse action would not have happened to a particular
plaintiff (not just a claim about a population of plaintiffs) means
showing not just causation but the more complex property called
\emph{actual causation}.
(What we are calling \emph{causation} is more precisely known as
\emph{type causation} when contrasting it with actual causation.)
Actual causation is formally identified using analytical methods
(e.g.,~\cite{halpern16book}).  Practically,
this relies on indirect ways of establishing causation such as verbal claims or
inconsistencies in applying rules~\cite{guerin-disparate-treatment}.

For these reasons, we only claim that the core essence of disparate
treatment is reduced to a question of causation.

\subsection{Indirect Discrimination: Frequentist Associative Independence}

Another form of discrimination happens when
a policy or system has disproportionate impacts upon one group by lacking
group parity.
We call this \emph{indirect discrimination}.
It is defined in a similar manner as \emph{indirect use privacy}:
the output is not statistically associated with the sensitive attribute $X$.
Absolute indirect nondiscrimination holds iff there is frequentist
associative independence:
for all $o$, $x_1$, $x_2$, and $B_P$,
\[ \Fr[\inter{O}\equals o \given \inter{X}\equals x_1, SE, B_P]
 = \Fr[\inter{O}\equals o \given \inter{X}\equals x_2, SE, B_P] \]
where
$SE = \{\inter{O}{=}\prog(\inter{X},\inter{A}), \inter{X}{=}X, \inter{A}{=}A\}$
as usual.
As before, $X$ is intuitively a protected attribute, such as gender or
race, instead of a private attribute.

Preventing indirect discrimination (ensuring group parity) is the
basis for a number technical approaches for nondiscrimination in
machine
learning~\cite{zliobaite11,kamishima12euromlkdd,zemel13icml,feldman15kdd}.

\paragraph{Disparate Impact.}

Whereas the disparate treatment form of illegal discrimination looked at
employment practices that directly treat protected groups
differently from others, disparate impact is a form of illegal
discrimination that can arise when even facially neutral policies lead
to indirect discrimination.
In employment law, the 80~percent rule~\cite{burger71scotus} is used to test
whether an employer's selection system has a disproportionate impact on protected
groups by measuring the ratio of positive outcomes across groups. If the ratio
is less than $0.8$, then the employer may be called upon to explain this
disparity.

\begin{prp}
  The 80~percent rule is satisfied by a system iff it has $\epsilon$-frequentist associative independence for protected groups, where $e^{-\epsilon} = 0.8$.
\end{prp}

Failing the 80~percent rule does not, in itself, mean that the
employer will be found liable for disparate impact.
As with disparate treatment, only certain entities are covered and
exceptions exist.
For example, the employer will not be held liable if the practice with
a disparate impact is a \emph{business necessity}.
Thus, similarly, we do not claim that all of disparate impact can be
reduced to a statistical test, but its core concept can be.

\section{Results for Free}
\label{sec:transfers}

Having recognized the relationship between privacy and
nondiscrimination, we get theorems and methods from one for the other
at no additional cost. This motivates the case for the exchange
of techniques between the communities studying the two independently.
In particular, the
Dwork--Naor impossibility result for statistical disclosure~\cite{dwork08jpc} translates
to an impossibility result for fairness that states that there will always exist
a subpopulation for which a system violates group parity. Further, we identify
techniques in probabilistic program analysis and data privacy geared towards minimizing quantitative
information flow that can be used to enforce nondiscrimination.
Finally, we mention ideas for generalizing the definitions found in
this work that can apply to either value.

\subsection{Dwork--Naor's Impossibility Result}
\label{sec:imposs}

Dwork and Naor showed the impossibility of statistical nondisclosure
(associative inferential privacy) when releasing useful
outputs~\cite{dwork08jpc}.
That theorem will carry over for discrimination.

They write~\cite[p.~1]{dwork08jpc}:
\begin{quote}
The intuition behind the proof of impossibility is captured by the following parable. Suppose one's exact height were considered a sensitive piece of information, and that revealing the exact height of an individual were a privacy breach. Assume that the database yields the average heights of women of different nationalities. An adversary who has access to the statistical database and the auxiliary information ``Terry Gross is two inches shorter than the average Lithuanian woman'' learns Terry Gross' height, while anyone learning only the auxiliary information, without access to the average heights, learns relatively little.
\end{quote}
Generalizing this parable to sensitive conditions other than Terry
Gross' height and to computations other than the average, requires
characterizing the sensitive conditions, computations, and
adversaries for which such reasoning holds.
Dwork and Naor do so for a fairly general characterization.
We will do so for a narrower but simpler characterization, which
allows making the points we wish to make more straightforward.

To understand our characterization, first note that the reasoning in the
parable requires that the sensitive condition is not already known to
the adversary.
So a straightforward way of generalizing it would be
\begin{quote}
For all informative computations $s$, outputs $o$, and sensitive conditions $\varphi$,
there exists some background condition $B$ such that $\varphi$ is not known from $B$ and $SE$,
but $\varphi$ will become known upon seeing that the output is $o$.
\end{quote}

The auxiliary information that ``Terry Gross is two inches shorter than
the average Lithuanian woman'' can be generalized to $\inter{O}{=}o$
implies $\varphi$ (or $\neg\varphi$ if $\varphi$ is actually false).
However, this does not hold.
The reason is that such auxiliary knowledge along with $SE$ (which is
assumed to be known by all) can alone imply that $\varphi$ holds.
One reason this is possible is that $\varphi$ can just follow from $SE$.
Alternatively, $O{=}o$ might follow from $SE$.
More complicatedly, $\varphi$ might be known to hold for every output
other than $o$ from $SE$.
In this case, the auxiliary information covers the only case where
$\varphi$ was not known to hold, implying that $\varphi$ holds without
needing to see the output.
(Appendix~\ref{app:ce} provides such a case.)

Rather than consider the possibility that some more complex auxiliary
information could deal with this case, we instead prove a more
restricted theorem.  Intuitively, it is
\begin{quote}
For all informative computations $s$, outputs $o$, and
sensitive conditions $\varphi$ whose truth or falsehood is not simply known from any of the outputs,
there exists some auxiliary information $B$ such that $\varphi$ is not known from $B$ and $SE$,
but $\varphi$ will become known upon seeing that the output is $o$.
\end{quote}

Since our focus is showing how privacy and nondiscrimination are related,
rather prove this result for just an epistemic notion of probability,
we provide a generic proof working for frequentist or Bayesian probabilities.

To do so, we generalize auxiliary information to be any
\emph{background condition}.
We also introduce the terms \emph{closed} and \emph{open} to
generalize \emph{known} and \emph{unknown}.
Section~\ref{sec:results} provides the generic results while
Sections~\ref{sec:privacy-interpretation} and~\ref{sec:nondiscrimination-interpretation} discuss how to interpret these
results for privacy and nondiscrimination.

\subsubsection{Generic Results}
\label{sec:results}

We start by making the concepts mentioned above precise.

We again use $B$ for background information, but since we do not use
causal interventions for these results, $B$ is not limited background
variables.
Let $\varphi$ be \emph{closed} for $B$ iff
$B$ is consistent and $\Pr[\varphi \given B]$ is $0$ or $1$.
Since we are working with discrete probabilities, this is the same as $B$ either proving or disproving $\varphi$.
If $\varphi$ is closed for $B$, then $\neg\varphi$ is closed for $B$.

Let $\varphi$ be \emph{open} for $B$ iff $B$ is consistent and $0 < \Pr[\varphi \given B] < 1$.
If $\varphi$ is open for $B$, then $\neg\varphi$ is open for $B$.
If $\varphi$ is open for $B$, then $\both{\varphi}{B}$ is consistent since $\Pr[\varphi \given SE] > 0$.

For consistent $B$, $\varphi$ is either open for $B$ or closed for $B$, but not both.

We say that a system is \emph{uninformative} if there exists an
output $o$ such that $\Pr[\inter{O}{=}o \given SE] = 1$.
In this case, there is no point to observing its output since its
output was already known from $SE$.
A system might be uninformative because it is a constant
function or because only a single input pair is possible:
$\Pr[X{=}x \land A{=}a \given SE] = 1$ for some $x$ and $a$.
We say the system is \emph{informative} is it is not uninformative:
for all $o$, $\Pr[\inter{O}{=}o \given SE] < 1$.

We say that a system $s$ can \emph{trivially close} $\varphi$ given $SE$ iff
$\varphi$ is open for $SE$, but
$s$ can produce an output $o$ that shows that either $\varphi$ or $\neg\varphi$ holds.
That is, for some $o$,
$\varphi$ is open for $SE$ but closed for $\both{\inter{O}{=}o}{SE}$.
Even without considering the background information of an adversary,
such a system closes $\varphi$ for at least one output.

\begin{thm}\label{thm:imposs}
Consider $SE = \{\inter{O}{=}s(\inter{X},\inter{A}), \inter{X}{=}X, \inter{A}{=}A\}$.
For all conditions $\varphi$, one of the following is true:
 $\varphi$ is closed for $SE$,
 $s$ is uninformative,
 $s$ can trivially close $\varphi$ given $SE$, or
 for all $o$, $\varphi$ is open for $\both{\inter{O}{\neq}o \lor \varphi}{SE}$ but
  closed for $\both{\inter{O}{=}o, \inter{O}{\neq}o \lor \varphi}{SE}$.
\end{thm}

This result follows from a series of lemmas, starting with one about open propositions.

\begin{lem}\label{lem:open-disj}
If %
$\varphi$ is open for $\both{\psi}{B}$,
then $\varphi$ is open for $\both{\varphi \lor \psi}{B}$.
\end{lem}
\begin{proof}

We must show that, when the above condition holds,
\begin{align}
0 &< \Pr[\varphi \given \varphi \lor \psi, B] < 1\\
0 &< \frac{\Pr[\varphi \lor \psi \given \varphi, B] * \Pr[\varphi \given B]}{\Pr[\varphi \lor \psi \given B]} < 1\\
0 &< \frac{1 * \Pr[\varphi \given B]}{\Pr[\varphi \given B] + \Pr[\psi \given B] - \Pr[\varphi \land \psi \given B]} < 1\\
0 &< \frac{\Pr[\varphi \given B]}{\Pr[\varphi \given B] + \Pr[\psi \given B] - \Pr[\psi \given B] * \Pr[\varphi \given \psi, B]} < 1 \label{ln:open-disj-ts}
\end{align}

Since $\varphi$ is open for $\both{\psi}{B}$, $0 < \Pr[\varphi \given \psi, B] < 1$.
By Bayes rule,
\begin{align}
0 &< \Pr[\varphi \given \psi, B]\\
0 &< \Pr[\varphi \given B]*\frac{\Pr[\psi \given \varphi, B]}{\Pr[\psi \given B]}
\end{align}
which implies that $0 < \Pr[\varphi \given B]$.
Thus,
\begin{align}
0 &< \frac{\Pr[\varphi \given B]}{\Pr[\varphi \given B] + \Pr[\psi \given B] - \Pr[\psi \given B] * \Pr[\varphi \given \psi, B]}\label{ln:open-disj-g0}
\end{align}

Since %
$0 < \Pr[\varphi \given \psi, B] < 1$,
\begin{align}
\Pr[\psi \given B] * \Pr[\varphi \given \psi, B] &< \Pr[\psi \given B]
\end{align}
Thus,
\begin{align}
0 &< \Pr[\psi \given B] - \Pr[\psi \given B] * \Pr[\varphi \given \psi, B]\\
\Pr[\varphi \given B] &< \Pr[\psi \given B] - \Pr[\psi \given B] * \Pr[\varphi \given \psi, B] + \Pr[\varphi \given B]
\end{align}
and
\begin{align}
\frac{\Pr[\varphi \given B]}{\Pr[\varphi \given B] + \Pr[\psi \given B] - \Pr[\psi \given B] * \Pr[\varphi \given \psi, B]} &< 1\label{ln:open-disj-l1}
\end{align}
Together~\eqref{ln:open-disj-g0} and~\eqref{ln:open-disj-l1} show that~\eqref{ln:open-disj-ts} holds.
Thus, $\varphi$ is open for $\both{\varphi \lor \psi}{B}$.
\end{proof}

We next prove a theorem about systems $s$ that cannot trivially close $\varphi$.
\begin{lem}\label{lem:open-eq-open-neq}
Consider $SE = \{\inter{O}{=}s(\inter{X},\inter{A}), \inter{X}{=}X, \inter{A}{=}A\}$.
For all systems $s$ that are informative,
if for all $o$, $\varphi$ is open for
$\both{\inter{O}{=}o}{SE}$,
then
for all $o$, $\varphi$ is open for
$\both{\inter{O}{\neq}o}{SE}$.
\end{lem}
\begin{proof}
We must show that, when the above conditions hold, for all $o$,
\begin{align}
0 &< \Pr[\varphi \given \inter{O}{\neq}o, SE] <1\\
0 &< \frac{\Pr[\varphi \land \inter{O}{\neq}o \given SE]}{\Pr[\inter{O}{\neq}o \given SE]} <1\\
0 &< \frac{\Pr[\varphi \land \bigvee_{o' \in \mc{O} \st o' \neq o} \inter{O}{=}o' \given SE]}{\Pr[\bigvee_{o' \in \mc{O} \st o' \neq o} \inter{O}{=}o' \given SE]} <1\\
0 &< \frac{\sum_{o' \in \mc{O} \st o' \neq o} \Pr[\varphi \land \inter{O}{=}o' \given SE]}{\sum_{o' \in \mc{O} \st o' \neq o} \Pr[\inter{O}{=}o' \given SE]} <1 \label{ln:open-eq-open-neq-ts}
\end{align}

Since $\varphi$ is open, for all $o$, $\Pr[\varphi \given O{=}o', SE] < 1$.
Thus,
\begin{align}
\sum_{o' \in \mc{O} \st o' \neq o} \Pr[\varphi \land \inter{O}{=}o' \given SE]
&= \sum_{o' \in \mc{O} \st o' \neq o} \Pr[\varphi \given \inter{O}{=}o', SE]*\Pr[\inter{O}{=}o' \given SE]\\
&< \sum_{o' \in \mc{O} \st o' \neq o} \Pr[\inter{O}{=}o' \given SE]
\end{align}
which implies that
\begin{align}
\frac{\sum_{o' \in \mc{O} \st o' \neq o} \Pr[\varphi \land \inter{O}{=}o' \given SE]}{\sum_{o' \in \mc{O} \st o' \neq o} \Pr[\inter{O}{=}o' \given SE]} &< 1 \label{ln:open-eq-open-neq-l1}
\end{align}

Since $s$ is informative,
for all outputs $o'$,
$\Pr[\inter{O}{=}o' \given SE] < 1$.
Thus,
there must exist at least two outputs $o_1$ and $o_2$ such that
$0 < \Pr[\inter{O}{=}o_i \given SE]$ for $i$ in $\{1,2\}$.
Since, for all $o$, $\varphi$ is open for
$\both{\inter{O}{=}o}{SE}$,
for $o_i \in \{o_1, o_2\}$,
$0 < \Pr[\varphi \given \inter{O}{=}o_i, SE]$.
Since
$0 < \Pr[\varphi \given \inter{O}{=}o_i, SE]$
and
$0 < \Pr[\inter{O}{=}o_i \given SE]$,
\begin{align}
0 &< \Pr[\varphi \given \inter{O}{=}o_i, SE]*\Pr[\inter{O}{=}o_i \given SE]
  = \Pr[\varphi \land \inter{O}{=}o_i \given SE]
\end{align}
This implies that
\begin{align}
0
&< \sum_{o' \in \mc{O} \st o' \neq o} \Pr[\varphi \land \inter{O}{=}o' \given SE]
\end{align}
since at most of one $o_1$ and $o_2$ can be $o$ and for all $o'$, $0 \leq \Pr[\varphi \land \inter{O}{=}o' \given SE]$.
Thus,
\begin{align}
0
&< \frac{\sum_{o' \in \mc{O} \st o' \neq o} \Pr[\varphi \land \inter{O}{=}o' \given SE]}{\sum_{o' \in \mc{O} \st o' \neq o} \Pr[\inter{O}{=}o' \given SE]} \label{ln:open-eq-open-neq-g0}
\end{align}

The inequalities~\eqref{ln:open-eq-open-neq-l1}
and~\eqref{ln:open-eq-open-neq-g0} together imply
that~\eqref{ln:open-eq-open-neq-ts} holds.
Thus, $\varphi$ is open for $\both{\inter{O}{\neq}o}{SE}$.
\end{proof}

Theorem~\ref{thm:imposs} follows directly from the next lemma.
\begin{lem}\label{lem:imposs}
Consider $SE = \{\inter{O}{=}s(\inter{X},\inter{A}), \inter{X}{=}X, \inter{A}{=}A\}$
such that $s$ is informative.
For all conditions $\varphi$ such that
 $\varphi$ is open for $SE$ and
 $s$ cannot trivially close $\varphi$ given $SE$,
for all outputs $o$,
$\varphi$ is open for $\both{\inter{O}{\neq}o \lor \varphi}{SE}$ but closed for $\both{\inter{O}{=}o, \inter{O}{\neq}o \lor \varphi}{SE}$.
\end{lem}
\begin{proof}

Since
$\varphi$ is open for $SE$ and
$s$ cannot trivially close $\varphi$ given $SE$,
$\varphi$ is open for $\both{\inter{O}{=}o'}{SE}$
for all $o'$ including $o$.
Thus, since $s$ is informative,
Lemma~\ref{lem:open-eq-open-neq} applies and
$\varphi$ is open for $\both{\inter{O}{\neq}o'}{SE}$
for all $o'$ including $o$.
Since
$\varphi$ is open for $\both{\inter{O}{\neq}o}{SE}$,
Lemma~\ref{lem:open-disj} applies and
$\varphi$ is open for $\both{\varphi \lor \inter{O}{\neq}o}{SE}$.

On the other hand,
\begin{align}
\Pr[\varphi \given \varphi \lor \inter{O}{\neq}o, \inter{O}{=}o, SE]
&= \Pr[\varphi \given \varphi, \inter{O}{=}o, SE]
= 1
\end{align}
Thus, $\varphi$ is closed for $\both{\inter{O}{=}o, \inter{O}{\neq}o \lor \varphi}{SE}$.
\end{proof}

We can instantiate Theorem~\ref{thm:imposs} for either privacy or
nondiscrimination, as we do so below.

\subsubsection{Privacy Interpretation}
\label{sec:privacy-interpretation}

For privacy, Theorem~\ref{thm:imposs} corresponds to saying that there
always exist an adversary who will prevent inferential privacy of the
statistical nondisclosure form.

More precisely, for privacy,
let \emph{resolved} from $B$ mean known to be true or known to be
false from just $B$.
We say that someone \emph{learns} a proposition $\varphi$ from an
observation if $\varphi$ is unresolved from that entity's background
knowledge but is resolved from the entity's background knowledge with
observation added to it.
We say that $s$ \emph{can trivially resolve} $\varphi$ given $SE$ if
$\varphi$ is unresolved for $SE$ but is for the combination of an
output of $s$ and $SE$.

We can restate Theorem~\ref{thm:imposs} as
\begin{cor}\label{cor:imposs-privacy}
Consider $SE = \{\inter{O}{=}s(\inter{X},\inter{A}), \inter{X}{=}X, \inter{A}{=}A\}$.
For all propositions $\varphi$, one of the following is true:
 $\varphi$ is resolved from $SE$,
 $s$ is uninformative,
 $s$ can trivially resolve $\varphi$ given $SE$, or
 for all $o$, $\varphi$ is unresolved for $\both{\inter{O}{\neq}o \lor \varphi}{SE}$ but
  resolved for $\both{\inter{O}{=}o, \inter{O}{\neq}o \lor \varphi}{SE}$.
\end{cor}

Corollary~\ref{cor:imposs-privacy} implies the following:
\begin{cor}\label{cor:imposs-statistical-disclosure}
Consider $SE = \{\inter{O}{=}s(\inter{X},\inter{A}), \inter{X}{=}X, \inter{A}{=}A\}$.
For all propositions $\varphi$, if
 $\varphi$ is unresolved from $SE$ and
 $s$ is informative, then
 there exists an adversary that experiences a statistical disclosure
 for $\varphi$ upon seeing some output of $s$.
\end{cor}
\begin{proof}
Recall that a statistical disclosure happens for an adversary with
background knowledge $B$ if
$\Cr[\varphi \given O\equals o, B] \neq \Cr[\varphi \given B]$,
which will happen if $\varphi$ goes from unresolved to resolved.

Since
 $\varphi$ is unresolved from $SE$ and
 $s$ is informative,
by Corollary~\ref{cor:imposs-privacy},
either
 $s$ can trivially resolve $\varphi$ given $SE$ or
 for all $o$, $\varphi$ is unresolved for $\both{\inter{O}{\neq}o \lor \varphi}{SE}$ but
  resolved for $\both{\inter{O}{=}o, \inter{O}{\neq}o \lor \varphi}{SE}$.
In the first case, $\varphi$ goes from unresolved for $SE$ to
resolved for $\both{O{=}o}{SE}$ for some output $o$.
In the second case, $\varphi$ goes from unresolved for $\both{\inter{O}{\neq}o \lor \varphi}{SE}$ but
to resolved for $\both{\inter{O}{=}o, \inter{O}{\neq}o \lor \varphi}{SE}$ for any output $o$.
Either way, there is a statistical disclosure for some output $o$.
\end{proof}

Thus, it is impossible to have a system that both produces interesting
outputs and always prevents statistical disclosures.
Dwork and Naor's impossibility result~\cite{dwork08jpc}, which is even
stronger, can be viewed as a justification for focusing on the weaker
property of differential privacy.

\subsubsection{Nondiscrimination Interpretation}
\label{sec:nondiscrimination-interpretation}

For discrimination, the generic result corresponds to saying that
a lack of group parity (disparate impact) will always exist for some
subpopulation.

In more detail, for nondiscrimination, we say that a population
\emph{lacks all diversity} of an attribute $\varphi$ if either
everyone in that population has $\varphi$ or everyone in the
population lacks $\varphi$.
Otherwise, we say that the population \emph{has some diversity} of
$\varphi$.
Given a population $B_P$, let its \emph{$\psi$ subpopulation} be the
subpopulation of $B_P$ identified by keeping only those members of $B_P$
who have the attribute $\psi$.
We say that a population \emph{loses diversity} for $\varphi$ in its
$\psi$ subpopulation if the population has some diversity of $\varphi$ but
its $\psi$ subpopulation lacks all diversity of $\varphi$.
We say that $s$ \emph{can trivially remove diversity} for $\varphi$
from $SE$ if for some output $o$ of $s$, $SE$ loses diversity for
$\varphi$ in its $O{=}o$ subpopulation.

We can restate Theorem~\ref{thm:imposs} as
\begin{cor}\label{cor:imposs-nondisc}
Consider $SE = \{\inter{O}{=}s(\inter{X},\inter{A}), \inter{X}{=}X, \inter{A}{=}A\}$.
For all attributes $\varphi$, one of the following is true:
 $\varphi$ lack all diversity for $SE$,
 $s$ is uninformative,
 $s$ can trivially remove diversity for $\varphi$ from $SE$, or
 for all outputs $o$,
  the $\both{\inter{O}{\neq}o \lor \varphi}{SE}$ subpopulation of $SE$
  loses diversity of $\varphi$ in its $\inter{O}{=}o$ subpopulation.
\end{cor}

An example will make the consequences of this result more clear.
Consider a system $s$ that decided who to hire with the output $o$ meaning
hired and consider the attribute of being male $\varphi$.
Suppose that the system does not simply hire or not hire everyone,
meaning that it is informative.
Further suppose that it hires some but not all males, meaning it does
not trivially remove diversity for $\varphi$.
Then, there exists some subpopulation of all applicants such that that
subpopulation contains both males and females, but only the males are
hired from it.
In particular, the subpopulation of men and non-hired women
will go from
having diversity to lacking it when focusing on just the hired ones.
The consequence of this is that while we can demand that the system
hires both males and females, we cannot expect this to hold for all
subpopulations.

Given the contrived nature of the subpopulation of men and non-hired women,
this may seem obvious and uninteresting.
However, three points are worth bearing in mind.

First, this result shows that the absence of disparate impact for all
subpopulations cannot be met for interesting systems $s$.
Some bounds on the subpopulations considered must be placed on demands
for group parity, and future work can search for reasonable bounds
instead of trying to provide mechanisms providing universal
nondiscrimination.
This is similar in spirit to Dwork and Naor's result in that the
adversary was not particularly realistic there either, but the result
still forced researchers to give up the search for mechanisms ensuring
statistical nondisclosure in all cases.

Second, if any non-empty subpopulation of the
$\both{\inter{O}{\neq}o \lor \varphi}{SE}$ subpopulation can be
identified using some other attribute $\chi$, then it may seem that
the attribute $\chi$ is causing some form of intersectional
discrimination against those identified by $\varphi \land \chi$ for
attribute $\psi$ of getting hired ($\inter{O}{=}o$).
For example, suppose that all the applicants who are black falls into
the set of men and non-hired women, that is, no black women were
hired.
It is possible that some form of intersectional discrimination is at play,
that is, discrimination not against women per se, nor black people per
se, but rather against black women.

It is also possible that outcome came about by noise.
The possibility of noise looks more likely as the third attribute $\chi$
gets more contrived looking.
For example, if instead of it being a simple racial attribute,
consider $\chi$ identifying applicants who are from Spain and over the
age of 40.
Is there intersectional discrimination against older Spanish women?
Or is it just a fluke?

While we do not do so here, we conjecture that modeling the input
space $\mc{A}$ more explicitly as a rich and diverse set of attributes
would allow us to show that one can always identify some such $\chi$
that is simple.

Third, returning to the goal of this paper, our real aim was to show
that one can produce general results that apply to both privacy and
nondiscrimination and then specialize them for each.
Corollaries~\ref{cor:imposs-privacy} and~\ref{cor:imposs-nondisc}
meets this goal.

\subsection{Probabilistic Program Analysis}

There exists a significant body of work in the formal security literature on
achieving formal bounds on probabilistic properties, aimed at verifying
probabilistic security properties of programs.
Techniques used to estimate the
probability of outcomes include abstract
interpretation~\cite{monniaux2000abstract}, volume
computation~\cite{sankar2013static}, and sampling~\cite{claret2013dataflow}.
While this body of work has largely focused on security properties, as the
properties are equivalent at a mathematical level, techniques for the analysis
of probabilistic programs can be directly used for fairness properties.

In one
example, Albarghouti~et~al.~\cite{albarghouti2017bias} use a volume computation
technique to either verify that the program satisfies a probabilistic fairness
property or provides a counterexample that demonstrates the program is biased.
They focus on disparate impact as a notion of discrimination:
\[ \Pr[\inter{O}{=}o\given X{=}x_1] \leq k * \Pr[\inter{O}{=}o\given X{=}x_2] \]

Their approach requires the specification of two programs: $\mathtt{popModel}$ and $\mathtt{dec}$.
The program $\mathtt{dec}$ is the program that decides the outcomes $\inter{O}$ for individuals, and $\mathtt{popModel}$ generates random individuals from a population which are fed to $\mathtt{dec}$.
The problem of estimating disparate impact, then reduces to estimating the quantities $\Pr[\inter{O}{=}o \land X{=}x_i]$ and $\Pr[X{=}x_i]$ for $i \in \{1, 2\}$. All of these
quantities can be estimated from the outcome of the composed program $\mathtt{dec}\circ\mathtt{popModel}$ using methods from security.

As another example of the connection between privacy and
nondiscrimination, we note that probabilistic coupling forms the core
of approaches to ensuring both.
On the one hand, a line of work has used probabilistic couplings to
verify that an algorithm provides differential
privacy~\cite{barthe16lics,barthe16ccs,barthe17popl,albarghouthi17popl}.
On the other hand, Friedler~et~al.\@ have presented a series of
probabilistic notions of (un)fairness~\cite{friedler16arxiv}.
Their definitions \emph{structural bias} (their Def.~3.5),
\emph{direct discrimination} (their Def.~3.6),
and \emph{nondiscrimination} (their Def.~3.7)
each have couplings (their Def.~2.5) at their cores.

\subsection{Statistical Scrubbing of Inputs}

A separate body of work
\cite{salamatian2013elephant,makhdoumi2013privacy,calmon2012privacy,salamatian2015managing}
addresses inference attacks in data sets by mapping each row $D$ to a row $D'$
such that $Y$ is independent to sensitive attributes with respect to some fixed
background knowledge $B$. The key idea is that if the inputs $D$ are independent
to sensitive attributes $X$, $D'\bot X \given B$, then for any function $f$, $f(D') \bot X \given B$. We term such mappings from $D$ to $D'$ statistical scrubbing.

Independently, a body of
work~\cite{zemel13icml,feldman15kdd} in the machine learning literature
attempts to solve the same problem of statistical scrubbing in a data set to
achieve independence with respect to sensitive attributes for providing
group parity.

\subsection{Generalized Distance Metrics for Differential Privacy}

One idea that appears in both nondiscrimination and privacy literature is the
generalization of differential privacy from indistinguishability over neighboring databases
to indistinguishability over nearby inputs as measured by some distance metric. In~\cite{dwork12itcs}, Dwork~et~al.\@ propose a notion of fairness that states that similar people
should be treated similarly. Formally the definition requires that the distance in distribution over outputs be bounded by the distance between output, a classical Lipschitz continuity requirement. In~\cite{chatzikokolakis13pets}, Chatzikokolakis~et~al.\@ propose a similar generalization for different privacy notions that and apply the definition in the contexts
of smart meter privacy and geo-location privacy.

\section{Discussion}

We have explored the relationship between privacy and
nondiscrimination at a mathematical level.
We have shown that both use privacy and nondiscrimination definitions
come in associative and causal flavors.
We have shown, at a mathematical level, that the basic form of the two
associative definitions are identical and that the basic form of the
two causal definitions are identical.
This similarity has allowed us to show what the relationships are between both
associative definitions and both causal definitions at once.
It has also allowed us to re-use a proof about privacy to say
something about nondiscrimination.
We believe this work shows that, to avoid duplication of effort,
computer science research on privacy and nondiscrimination should not
be siloed into separate communities.

This is not to say that privacy and nondiscrimination are or should be
collapsed into one problem.
The definitions, despite having the same formal structures, differ in
their interpretations and what the variables model.
Neither can do the job of the other; attention must be paid to both.

Furthermore, privacy and nondiscrimination can interact both
positively and negatively.
For example, not disclosing one's race on an employment application
may help prevent discrimination in some settings.
Furthermore, nondiscrimination can lessen some privacy concerns.
For example, banning health insurers from discriminating based on
pre-existing conditions may lessen the degree of privacy patients with
expensive conditions seek.
On the other hand, privacy can hide discrimination.
Our work may be a starting point for exploring these interactions.

Lastly, in all cases, the mathematical characterizations we use are
simplifications of the the social norms and laws they characterize.
For one, both privacy and nondiscrimination norms contain exceptions.
An example for privacy is that patient--physician confidentiality may
be set aside by court orders.
As for nondiscrimination, disparate impact is allowed where a
``business necessity'' exists.
Furthermore, nondiscrimination law contains many features that we do
not capture, such as intent.
In particular, we replaced the vague concept of
\emph{intent} with the more precise concept of \emph{cause}.
Exploring what intent might look like for automated systems could
provide more nuanced models of disparate treatment.
Nevertheless, we believe our mathematical characterizations capture
the essence of what each norm or law is attempting to achieve.

\paragraph{Acknowledgements.}
We gratefully acknowledge funding support from
the National Science Foundation
(Grants 1514509, %
1704845, %
and 1704985). %
The opinions in this paper are those of the authors and do not
necessarily reflect the opinions of any funding sponsor or the United
States Government.

\bibliographystyle{IEEEtran}
\bibliography{no-edit}

\begin{thebibliography}{10}
\providecommand{\url}[1]{#1}
\csname url@samestyle\endcsname
\providecommand{\newblock}{\relax}
\providecommand{\bibinfo}[2]{#2}
\providecommand{\BIBentrySTDinterwordspacing}{\spaceskip=0pt\relax}
\providecommand{\BIBentryALTinterwordstretchfactor}{4}
\providecommand{\BIBentryALTinterwordspacing}{\spaceskip=\fontdimen2\font plus
\BIBentryALTinterwordstretchfactor\fontdimen3\font minus
  \fontdimen4\font\relax}
\providecommand{\BIBforeignlanguage}[2]{{%
\expandafter\ifx\csname l@#1\endcsname\relax
\typeout{** WARNING: IEEEtran.bst: No hyphenation pattern has been}%
\typeout{** loaded for the language `#1'. Using the pattern for}%
\typeout{** the default language instead.}%
\else
\language=\csname l@#1\endcsname
\fi
#2}}
\providecommand{\BIBdecl}{\relax}
\BIBdecl

\bibitem{dalenius77statistik}
T.~Dalenius, ``Towards a methodology for statistical disclosure control,''
  \emph{Statistik Tidskrift}, vol.~15, pp. 429--444, 1977.

\bibitem{warner65asa}
\BIBentryALTinterwordspacing
S.~L. Warner, ``Randomized response: A survey technique for eliminating evasive
  answer bias,'' \emph{Journal of the American Statistical Association},
  vol.~60, no. 309, pp. 63--69, 1965. [Online]. Available:
  \url{http://www.jstor.org/stable/2283137}
\BIBentrySTDinterwordspacing

\bibitem{dwork06crypto}
C.~Dwork, F.~Mcsherry, K.~Nissim, and A.~Smith, ``Calibrating noise to
  sensitivity in private data analysis,'' in \emph{Theory of Cryptography
  Conference}.\hskip 1em plus 0.5em minus 0.4em\relax Springer, 2006, pp.
  265--284.

\bibitem{dwork06icalp}
C.~Dwork, ``Differential privacy,'' in \emph{Automata, Languages and
  Programming, 33rd International Colloquium, {ICALP} 2006, Venice, Italy, July
  10--14, 2006, Proceedings, Part {II}}, ser. Lecture Notes in Computer
  Science, M.~Bugliesi, B.~Preneel, V.~Sassone, and I.~Wegener, Eds., vol.
  4052.\hskip 1em plus 0.5em minus 0.4em\relax Springer, 2006, pp. 1--12.

\bibitem{bassily13focs}
R.~Bassily, A.~Groce, J.~Katz, and A.~Smith, ``Coupled-worlds privacy:
  Exploiting adversarial uncertainty in statistical data privacy,'' in
  \emph{Proceedings of the 2013 IEEE 54th Annual Symposium on Foundations of
  Computer Science}.\hskip 1em plus 0.5em minus 0.4em\relax IEEE Computer
  Society, 2013, pp. 439--448.

\bibitem{kasiviswanathan14jpc}
S.~P. Kasiviswanathan and A.~Smith, ``On the `semantics' of differential
  privacy: A bayesian formulation,'' \emph{Journal of Privacy and
  Confidentiality}, vol.~6, no.~1, pp. 1--16, 2014.

\bibitem{mcsherry16blog1}
F.~McSherry, ``Lunchtime for data privacy,'' Blog:
  \url{https://github.com/frankmcsherry/blog/blob/master/posts/2016-08-16.md},
  2016.

\bibitem{mcsherry16blog2}
------, ``Differential privacy and correlated data,'' Blog:
  \url{https://github.com/frankmcsherry/blog/blob/master/posts/2016-08-29.md},
  2016.

\bibitem{tschantz17arxiv}
M.~C. Tschantz, S.~Sen, and A.~Datta, ``Differential privacy as a causal
  property,'' \emph{ArXiv}, vol. 1710.05899, 2017.

\bibitem{burger71scotus}
W.~Burger, ``Griggs v.\ duke power company,'' Opinion of the United States
  Supreme Court, 1971.

\bibitem{pearl09book}
J.~Pearl, \emph{Causality}, 2nd~ed.\hskip 1em plus 0.5em minus 0.4em\relax
  Cambridge University Press, 2009.

\bibitem{dwork08jpc}
C.~Dwork and M.~Naor, ``On the difficulties of disclosure prevention in
  statistical databases or the case for differential privacy,'' \emph{Journal
  of Privacy and Confidentiality}, vol.~2, no.~1, pp. 93--107, 2008.

\bibitem{datta2017use}
A.~Datta, M.~Fredrikson, G.~Ko, P.~Mardziel, and S.~Sen, ``Use privacy in
  data-driven systems: Theory and experiments with machine learnt programs,''
  in \emph{Proceedings of the 2017 ACM SIGSAC Conference on Computer and
  Communications Security}.\hskip 1em plus 0.5em minus 0.4em\relax ACM, 2017,
  pp. 1193--1210.

\bibitem{ghosh17itcs}
A.~Ghosh and R.~Kleinberg, ``Inferential privacy guarantees for differentially
  private mechanisms,'' in \emph{Proceedings of the 8th Innovations in
  Theoretical Computer Science Conference (ITCS 2017)}, 2017.

\bibitem{holland86jasa}
P.~W. Holland, ``Statistics and causal inference,'' \emph{Journal of the
  American Statistical Association}, vol.~81, pp. 945--970, 1986, with
  discussion.

\bibitem{gm82security}
J.~A. Goguen and J.~Meseguer, ``Security policies and security models,'' in
  \emph{Proceedings of the IEEE Symposium on Security and Privacy}, 1982, pp.
  11--20.

\bibitem{kifer14database}
D.~Kifer and A.~Machanavajjhala, ``Pufferfish: A framework for mathematical
  privacy definitions,'' \emph{ACM Trans. Database Syst.}, vol.~39, no.~1, pp.
  3:1--3:36, 2014.

\bibitem{dwork12itcs}
C.~Dwork, M.~Hardt, T.~Pitassi, O.~Reingold, and R.~Zemel, ``Fairness through
  awareness,'' in \emph{Proceedings of the 3rd Innovations in Theoretical
  Computer Science Conference}.\hskip 1em plus 0.5em minus 0.4em\relax ACM,
  2012, pp. 214--226.

\bibitem{dwork13slawr}
C.~Dwork and D.~K. Mulligan, ``It's not privacy, and it's not fair,''
  \emph{Stanford Law Review}, vol.~66, no.~1, pp. 35--41, 2013.

\bibitem{agan16tr}
A.~Y. Agan and S.~B. Starr, ``Ban the box, criminal records, and statistical
  discrimination: A field experiment,'' University of Michigan Law School, Law
  and Economics Research Paper Series 16-012, 2016.

\bibitem{strahilevitz2008privacy}
L.~J. Strahilevitz, ``Privacy versus antidiscrimination,'' \emph{U. Chi. L.
  Rev.}, vol.~75, p. 363, 2008.

\bibitem{tschantz15csf}
M.~C. Tschantz, A.~Datta, A.~Datta, and J.~M. Wing, ``A methodology for
  information flow experiments,'' in \emph{Computer Security Foundations
  Symposium}.\hskip 1em plus 0.5em minus 0.4em\relax IEEE, 2015.

\bibitem{mclean90sp}
J.~McLean, ``Security models and information flow,'' in \emph{Proceedings of
  the IEEE Computer Society Symposium on Research in Security and Privacy},
  1990, pp. 180--187.

\bibitem{mowbray92csf}
M.~Mowbray, ``Causal security,'' in \emph{Proceedings of the Computer Security
  Foundations Workshop}, 1992, pp. 54--62.

\bibitem{sewell00csf}
P.~Sewell and J.~Vitek, ``Secure composition of untrusted code: wrappers and
  causality types,'' in \emph{Computer Security Foundations Workshop, 2000.
  CSFW-13. Proceedings. 13th IEEE}, 2000, pp. 269--284.

\bibitem{kusner16aistats}
M.~J. Kusner, Y.~Sun, K.~Sridharan, and K.~Q. Weinberger, ``Private causal
  inference,'' in \emph{Proceedings of the 19th International Conference on
  Artificial Intelligence and Statistics}, ser. Proceedings of Machine Learning
  Research, A.~Gretton and C.~C. Robert, Eds., vol.~51.\hskip 1em plus 0.5em
  minus 0.4em\relax PMLR, 2016, pp. 1308--1317.

\bibitem{pcast14}
{President's Council of Advisors on Science and Technology}, ``Big data and
  privacy: A technological perspective,'' Executive Office of the President
  (USA), Report to the President, 2014.

\bibitem{kifer11sigmod}
D.~Kifer and A.~Machanavajjhala, ``No free lunch in data privacy,'' in
  \emph{Proceedings of the 2011 ACM SIGMOD International Conference on
  Management of data}.\hskip 1em plus 0.5em minus 0.4em\relax ACM, 2011, pp.
  193--204.

\bibitem{kifer12pods}
------, ``A rigorous and customizable framework for privacy,'' in
  \emph{Proceedings of the 31st ACM SIGMOD-SIGACT-SIGAI Symposium on Principles
  of Database Systems}.\hskip 1em plus 0.5em minus 0.4em\relax ACM, 2012, pp.
  77--88.

\bibitem{he14sigmod}
X.~He, A.~Machanavajjhala, and B.~Ding, ``Blowfish privacy: Tuning
  privacy-utility trade-offs using policies,'' in \emph{Proceedings of the ACM
  SIGMOD International Conference on Management of Data (SIGMOD 2014)}.\hskip
  1em plus 0.5em minus 0.4em\relax ACM, 2014.

\bibitem{chen14vldbj}
R.~Chen, B.~C. Fung, P.~S. Yu, and B.~C. Desai, ``Correlated network data
  publication via differential privacy,'' \emph{The VLDB Journal}, vol.~23,
  no.~4, pp. 653--676, 2014.

\bibitem{zhu15tifs}
T.~Zhu, P.~Xiong, G.~Li, and W.~Zhou, ``Correlated differential privacy: Hiding
  information in non-{IID} data set,'' \emph{IEEE Transactions on Information
  Forensics and Security}, vol.~10, no.~2, pp. 229--242, 2015.

\bibitem{liu16ndss}
C.~Liu, S.~Chakraborty, and P.~Mittal, ``Dependence makes you vulnerable:
  Differential privacy under dependent tuples,'' in \emph{Network and
  Distributed System Security Symposium (NDSS)}.\hskip 1em plus 0.5em minus
  0.4em\relax The Internet Society, 2016.

\bibitem{shannon49bell}
C.~E. Shannon, ``Communication theory of secrecy systems,'' \emph{Bell Labs
  Technical Journal}, vol.~28, no.~4, pp. 656--715, 1949.

\bibitem{ghosh16inferential}
A.~Ghosh and R.~Kleinberg, ``Inferential privacy guarantees for differentially
  private mechanisms,'' \emph{CoRR}, vol. abs/1603.01508, 2016.

\bibitem{alivim11icalp}
M.~Alvim, M.~Andr\'{e}s, K.~Chatzikokolakis, and C.~Palamidessi, ``On the
  relation between differential privacy and quantitative information flow,'' in
  \emph{38th International Colloquium on Automata, Languages and Programming --
  ICALP 2011}, ser. Lecture Notes in Computer Science, L.~Aceto, M.~Henzinger,
  and J.~Sgall, Eds., vol. 6756.\hskip 1em plus 0.5em minus 0.4em\relax
  Springer, 2011, pp. 60--76.

\bibitem{cuff16ccs}
P.~Cuff and L.~Yu, ``Differential privacy as a mutual information constraint,''
  in \emph{Proceedings of the 2016 ACM SIGSAC Conference on Computer and
  Communications Security}, ser. CCS '16.\hskip 1em plus 0.5em minus
  0.4em\relax ACM, 2016, pp. 43--54.

\bibitem{mcsherry17blog1}
F.~McSherry, ``On ``differential privacy as a mutual information
  constraint'','' Blog:
  \url{https://github.com/frankmcsherry/blog/blob/master/posts/2017-01-26.md},
  2017.

\bibitem{hardt16arxiv}
M.~Hardt, E.~Price, and N.~Srebro, ``Equality of opportunity in supervised
  learning,'' \emph{ArXiv}, vol. 1610.02413, 2016.

\bibitem{hardt16nips}
M.~Hardt, E.~Price, , and N.~Srebro, ``Equality of opportunity in supervised
  learning,'' in \emph{Advances in Neural Information Processing Systems 29},
  D.~D. Lee, M.~Sugiyama, U.~V. Luxburg, I.~Guyon, and R.~Garnett, Eds.\hskip
  1em plus 0.5em minus 0.4em\relax Curran Associates, Inc., 2016, pp.
  3315--3323.

\bibitem{kilbertus2017causal}
N.~Kilbertus, M.~Rojas-Carulla, G.~Parascandolo, M.~Hardt, D.~Janzing, and
  B.~Sch\"{o}lkopf, ``Avoiding discrimination through causal reasoning,'' in
  \emph{Proceedings from the conference "Neural Information Processing Systems
  2017.}\hskip 1em plus 0.5em minus 0.4em\relax Curran Associates, Inc., 2017,
  pp. 656--666.

\bibitem{kusner2017counterfactual}
M.~J. Kusner, J.~Loftus, C.~Russell, and R.~Silva, ``Counterfactual fairness,''
  in \emph{Advances in Neural Information Processing Systems 30}, I.~Guyon,
  U.~V. Luxburg, S.~Bengio, H.~Wallach, R.~Fergus, S.~Vishwanathan, and
  R.~Garnett, Eds.\hskip 1em plus 0.5em minus 0.4em\relax Curran Associates,
  Inc., 2017, pp. 4066--4076.

\bibitem{bonchi2017exposing}
F.~Bonchi, S.~Hajian, B.~Mishra, and D.~Ramazzotti, ``Exposing the
  probabilistic causal structure of discrimination,'' \emph{International
  Journal of Data Science and Analytics}, vol.~3, no.~1, pp. 1--21, 2017.

\bibitem{cowgill2017algorithmic}
B.~Cowgill and C.~Tucker, ``Algorithmic bias: A counterfactual perspective,''
  Tech. Rep., 2017.

\bibitem{calders10dmkd}
T.~Calders and S.~Verwer, ``Three naive {B}ayes approaches for
  discrimination-free classification,'' \emph{Data Mining and Knowledge
  Discovery}, vol.~21, no.~2, pp. 277--292, 2010.

\bibitem{kamishima2011}
T.~Kamishima, S.~Akaho, and J.~Sakuma, ``Fairness-aware learning through
  regularization approach,'' in \emph{Proceedings of the 2011 {IEEE} 11th
  International Conference on Data Mining Workshops (ICDMW 2011)}, 2011, pp.
  643--650.

\bibitem{zemel13icml}
\BIBentryALTinterwordspacing
R.~Zemel, Y.~Wu, K.~Swersky, T.~Pitassi, and C.~Dwork, ``Learning fair
  representations,'' in \emph{Proceedings of the 30th International Conference
  on Machine Learning (ICML-13)}, S.~Dasgupta and D.~Mcallester, Eds.,
  vol.~28.\hskip 1em plus 0.5em minus 0.4em\relax JMLR Workshop and Conference
  Proceedings, May 2013, pp. 325--333. [Online]. Available:
  \url{http://jmlr.org/proceedings/papers/v28/zemel13.pdf}
\BIBentrySTDinterwordspacing

\bibitem{feldman15kdd}
M.~Feldman, S.~A. Friedler, J.~Moeller, C.~Scheidegger, and
  S.~Venkatasubramanian, ``Certifying and removing disparate impact,'' in
  \emph{Proceedings of the 21th ACM SIGKDD International Conference on
  Knowledge Discovery and Data Mining}.\hskip 1em plus 0.5em minus 0.4em\relax
  ACM, 2015, pp. 259--268.

\bibitem{kleinberg17itcs}
J.~Kleinberg, S.~Mullainathan, and M.~Raghavan, ``Inherent trade-offs in the
  fair determination of risk scores,'' in \emph{Innovations in Theoretical
  Computer Science}, 2017.

\bibitem{chouldechova16fatml}
A.~Chouldechova, ``Fair prediction with disparate impact: A study of bias in
  recidivism prediction instruments,'' \emph{ArXiv}, vol. 1610.07524, 2016,
  presented at the 3rd Workshop on Fairness, Accountability, and Transparency
  in Machine Learning, 2016.

\bibitem{mironov17csf}
I.~Mironov, ``Rényi differential privacy,'' in \emph{2017 IEEE 30th Computer
  Security Foundations Symposium (CSF)}, 2017, pp. 263--275.

\bibitem{clarkson05csf}
M.~Clarkson, A.~Myers, and F.~Schneider, ``Belief in information flow,'' in
  \emph{Computer Security Foundations, 2005. CSFW-18 2005. 18th IEEE Workshop},
  2005, pp. 31--45.

\bibitem{halpern08tissec}
J.~Y. Halpern and K.~R. O'Neill, ``Secrecy in multiagent systems,'' \emph{ACM
  Trans. Inf. Syst. Secur.}, vol.~12, no.~1, pp. 5:1--5:47, 2008.

\bibitem{datta17ccs}
A.~Datta, M.~Fredrikson, G.~Ko, P.~Mardziel, and S.~Sen, ``Use privacy in
  data-driven systems,'' in \emph{Proceedings of 24th ACM Conference on
  Computer and Communications Security, October 2017}, 2007.

\bibitem{stewart77scotus}
P.~Stewart, ``{International Brotherhood of Teamsters v.\ United States},''
  Opinion of the United States Supreme Court, 1977.

\bibitem{dennett87book}
D.~C. Dennett, \emph{The Intentional Stance}.\hskip 1em plus 0.5em minus
  0.4em\relax MIT Press/A Bradford Book, 1987.

\bibitem{romei14survey}
A.~Romei and S.~Ruggieri, ``A multidisciplinary survey on discrimination
  analysis,'' \emph{The Knowledge Engineering Review}, vol.~29, pp. 582--638,
  2014.

\bibitem{united2008ninth}
U.~S.~C. of~Appeals (9th Circuit). Committee~on Model Civil Jury~Instructions,
  \emph{Ninth Circuit Manual of Model Jury Instructions: Civil}.\hskip 1em plus
  0.5em minus 0.4em\relax West Publishing, 2008.

\bibitem{halpern16book}
J.~Y. Halpern, \emph{Actual Causality}.\hskip 1em plus 0.5em minus 0.4em\relax
  MIT Press, 2016.

\bibitem{guerin-disparate-treatment}
``Disparate treatment discrimination,'' Legal Encyclopedia at
  \url{https://www.nolo.com/legal-encyclopedia/disparate-treatment-discrimination.html},
  2017.

\bibitem{zliobaite11}
I.~Žliobaitė, F.~Kamiran, and T.~Calders, ``Handling conditional
  discrimination,'' in \emph{Proceedings of the 2011 IEEE 11th International
  Conference on Data Mining}.\hskip 1em plus 0.5em minus 0.4em\relax IEEE
  Computer Society, 2011, pp. 992--1001.

\bibitem{kamishima12euromlkdd}
T.~Kamishima, S.~Akaho, H.~Asoh, and J.~Sakuma, ``Fairness-aware classifier
  with prejudice remover regularizer,'' in \emph{Proceedings of the 2012
  European Conference on Machine Learning and Knowledge Discovery in Databases
  - Volume Part II}.\hskip 1em plus 0.5em minus 0.4em\relax Springer-Verlag,
  2012, pp. 35--50.

\bibitem{monniaux2000abstract}
D.~Monniaux, ``Abstract interpretation of probabilistic semantics,'' in
  \emph{Proceedings of the 7th International Symposium on Static
  Analysis}.\hskip 1em plus 0.5em minus 0.4em\relax Springer-Verlag, 2000, pp.
  322--339.

\bibitem{sankar2013static}
S.~Sankaranarayanan, A.~Chakarov, and S.~Gulwani, ``Static analysis for
  probabilistic programs: Inferring whole program properties from finitely many
  paths,'' in \emph{Proceedings of the 34th ACM SIGPLAN Conference on
  Programming Language Design and Implementation}.\hskip 1em plus 0.5em minus
  0.4em\relax ACM, 2013, pp. 447--458.

\bibitem{claret2013dataflow}
G.~Claret, S.~K. Rajamani, A.~V. Nori, A.~D. Gordon, and J.~Borgstr\"{o}m,
  ``Bayesian inference using data flow analysis,'' in \emph{Proceedings of the
  2013 9th Joint Meeting on Foundations of Software Engineering}.\hskip 1em
  plus 0.5em minus 0.4em\relax ACM, 2013, pp. 92--102.

\bibitem{albarghouti2017bias}
A.~Albarghouthi, L.~D'Antoni, S.~Drews, and A.~V. Nori, ``Quantifying program
  bias,'' \emph{ArXiv}, vol. 1702.05437, 2017.

\bibitem{barthe16lics}
G.~Barthe, M.~Gaboardi, B.~Gr{\'e}goire, J.~Hsu, and P.-Y. Strub, ``Proving
  differential privacy via probabilistic couplings,'' in \emph{Proceedings of
  the 31st Annual ACM/IEEE Symposium on Logic in Computer Science}, ser. LICS
  '16.\hskip 1em plus 0.5em minus 0.4em\relax ACM, 2016, pp. 749--758.

\bibitem{barthe16ccs}
G.~Barthe, N.~Fong, M.~Gaboardi, B.~Gr{\'e}goire, J.~Hsu, and P.-Y. Strub,
  ``Advanced probabilistic couplings for differential privacy,'' in
  \emph{Proceedings of the 2016 ACM SIGSAC Conference on Computer and
  Communications Security}, ser. CCS '16.\hskip 1em plus 0.5em minus
  0.4em\relax ACM, 2016, pp. 55--67.

\bibitem{barthe17popl}
G.~Barthe, T.~Espitau, B.~Gr{\'e}goire, J.~Hsu, and P.-Y. Strub, ``Proving
  expected sensitivity of probabilistic programs,'' \emph{Proc. ACM Program.
  Lang.}, vol.~2, no. POPL, pp. 57:1--57:29, 2017.

\bibitem{albarghouthi17popl}
A.~Albarghouthi and J.~Hsu, ``Synthesizing coupling proofs of differential
  privacy,'' \emph{Proc. ACM Program. Lang.}, vol.~2, no. POPL, pp.
  58:1--58:30, 2017.

\bibitem{friedler16arxiv}
S.~A. Friedler, C.~Scheidegger, and S.~Venkatasubramanian, ``On the
  (im)possibility of fairness,'' \emph{ArXiv}, vol. 1609.07236, 2016.

\bibitem{salamatian2013elephant}
S.~Salamatian, A.~Zhang, F.~d.~P.~Calmon, S.~Bhamidipati, N.~Fawaz, B.~Kveton,
  P.~Oliveira, and N.~Taft, ``How to hide the elephant- or the donkey- in the
  room: Practical privacy against statistical inference for large data,'' in
  \emph{2013 IEEE Global Conference on Signal and Information Processing},
  2013, pp. 269--272.

\bibitem{makhdoumi2013privacy}
A.~Makhdoumi and N.~Fawaz, ``Privacy-utility tradeoff under statistical
  uncertainty,'' in \emph{2013 51st Annual Allerton Conference on
  Communication, Control, and Computing (Allerton)}, 2013, pp. 1627--1634.

\bibitem{calmon2012privacy}
F.~du~Pin~Calmon and N.~Fawaz, ``Privacy against statistical inference,'' in
  \emph{2012 50th Annual Allerton Conference on Communication, Control, and
  Computing (Allerton)}, 2012, pp. 1401--1408.

\bibitem{salamatian2015managing}
S.~Salamatian, A.~Zhang, F.~du~Pin~Calmon, S.~Bhamidipati, N.~Fawaz, B.~Kveton,
  P.~Oliveira, and N.~Taft, ``Managing your private and public data: Bringing
  down inference attacks against your privacy,'' \emph{IEEE Journal of Selected
  Topics in Signal Processing}, vol.~9, no.~7, pp. 1240--1255, 2015.

\bibitem{chatzikokolakis13pets}
K.~Chatzikokolakis, M.~E. Andr{\'e}s, N.~E. Bordenabe, and C.~Palamidessi,
  \emph{Broadening the Scope of Differential Privacy Using Metrics}.\hskip 1em
  plus 0.5em minus 0.4em\relax Springer Berlin Heidelberg, 2013, pp. 82--102.

\end{thebibliography}

\appendix

\section{Example System}
\label{app:ce}

Here we present an example system showing that $\varphi$ could be
closed for $\both{O{\neq}o \lor \varphi}{SE}$ despite not being closed
for $SE$.
This motivates the condition that the system $s$ cannot trivially
close $\varphi$ given $SE$ in Lemma~\ref{lem:imposs}.
In particular, note that example below does trivially close $\varphi$
given $SE$ when the output is $\msf{nonpositive}$.

As before, let $SE$ be $\inter{O}{=}s(\inter{X},\inter{A})$, $\inter{X}{=}X$, and $\inter{A}{=}A$, and let it hold.
Consider the program $s(x,a) = \msf{pos}(x,a)$ which returns whether $x$ is positive, with $\msf{positive}$ meaning yes it is and $\msf{nonpositive}$ meaning no it is not.
Since $\msf{pos}$ ignores its second input, we will drop it.
Suppose that the range $\mc{X}$ of $X$ is $\{0, 1, 2\}$.

Consider the sensitive condition $\varphi$ that is whether $X$ is even, that is, $\msf{even}(X)$.
Seeing the output $\inter{O}$ only sometimes reveals whether $\msf{even}(X)$ holds.
If the output is $\msf{pos}(\inter{X}) = \msf{nonpositive}$, then the input must have been $0$, and $\msf{even}(X)$ must hold.
If the output is $\msf{pos}(\inter{X}) = \msf{positive}$, then the input could have been $1$ or $2$, and $\msf{even}(X)$ may or may not hold.
The following diagram summarizes this state of affairs:
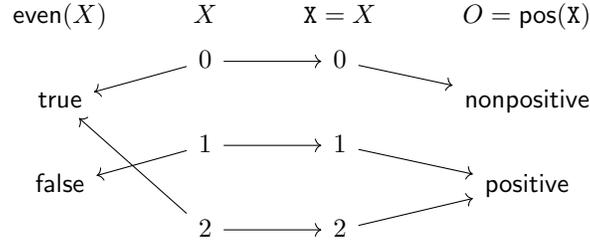
\begin{center}
\begin{tikzcd}[math mode=true, row sep=0.1em]%
\msf{even}(X) & X                      & \inter{X} = X & O = \msf{pos}(\inter{X}) \\
              & 0\arrow[ld]\arrow[r]   & 0\arrow[rd]   &                          \\
\msf{true}    &                        &               & \msf{nonpositive}       \\
              & 1\arrow[ld]\arrow[r]   & 1\arrow[rd]   &                          \\
\msf{false}   &                        &               & \msf{positive}           \\
              & 2\arrow[luuu]\arrow[r] & 2\arrow[ru]   &                          \\
\end{tikzcd}
\end{center}

To understand how this relates to Lemma~\ref{lem:imposs},
let us focus on when the output $\inter{O}$ takes on the value $o = \msf{positive}$.
Note that $\varphi$ is open for $\both{\inter{O}{=}\msf{positive}}{SE}$.

Consider the background context
$\inter{O}{\neq}\msf{positive} \lor \msf{even}(X)$, that is, that
$\inter{O}{=}\msf{positive}$ implies $\msf{even}(X)$.
Given this context, seeing the output $\msf{positive}$ would imply $\msf{even}(X)$.
Thus, given this context, one can learn the sensitive condition for either output of $s$.

However, one would not even need to see the output to learn the sensitive condition given this context.
Consider an analysis by cases.
If $\inter{O}{\neq}\msf{positive}$ holds, then $X$ must be $0$ and $\msf{even}(X)$ holds.
If $\msf{even}(X)$ holds, then $\msf{even}(X)$ must hold.
So, either way $\msf{even}(X)$ must hold.

The issue is that $\inter{O}$ not having the value $\msf{positive}$
implied too much: it implies that the output must be
$\msf{nonpositive}$, which implies $\varphi$ holds.
The following table, which shows the values of $\inter{X}$ that lead
to each value for $\varphi$ and $\inter{O}$, illustrates this issue:
\begin{center}
\begin{tab}{lll}
                       & \multicolumn{2}{c}{$\msf{even}(X)$} \\
\cmidrule(l){2-3}
$\msf{pos}(\inter{X})$ & $\msf{true}$ & $\msf{false}$\\
$\msf{nonpositive}$    & 0 &  \\
$\msf{positive}$       & 2 & 1\\
\end{tab}
\end{center}
Note that no value is both odd and nonpositive, meaning that $s$
trivially closes $\varphi$ for the nonpositive output.
This means that set of value of $\inter{X}$ that makes $\inter{O}{\neq}\msf{positive} \lor \varphi$ true, which is
equivalent to $\inter{O}{=}\msf{nonpositive} \lor \varphi$, only has
a value of $\inter{X}$ for when $\varphi$ is true.
Thus, it implies the truth of $\varphi$ without needing to see the
value of $\inter{O}$.

\end{document}